\documentclass[sigconf]{acmart}
\usepackage{algorithm}
\usepackage{algorithmic}
\usepackage{graphicx}   
\graphicspath{{figures/}}
\AtBeginDocument{%
  }

\acmYear{}\copyrightyear{}
\acmConference[]{}{}
\acmBooktitle{}
\acmDOI{}
\acmISBN{}
\title{Human-in-the-loop Learning Through Decentralized Communication Mechanisms}

\author{Yiting Hu}
\affiliation{%
  \institution{Singapore University of Technology and Design}
  \country{Singapore}}
\email{yiting_hu@mymail.sutd.edu.sg}

\author{Lingjie Duan}
\authornote{Corresponding author.}
\affiliation{%
  \institution{SUTD (current); moving to Hong
Kong University of Science and
Technology (Guangzhou) in late 2025}
\country{ }}
\email{lingjieduan@hkust-gz.edu.cn}

\thanks{This work is to appear in Mobihoc 2025. This work is supported in part by SUTD Kickstarter Initiative (SKI) Grant with no. SKI 2021\_04\_07; and in part by the Joint SMU-SUTD Grant with no. 22-LKCSB-SMU-053.}
\setcopyright{none}
\renewcommand\footnotetextcopyrightpermission[1]{}
\renewcommand{\shortauthors}{Hu et al.}
\begin{document}
\begin{abstract}
Information sharing platforms like TripAdvisor and Waze involve human agents as both information producers and consumers. All these platforms operate in a centralized way to collect agents' latest observations of new options (e.g., restaurants, hotels, travel routes) and share such information with all in real time. However, after hearing the central platforms' live updates, many human agents are found selfish and unwilling to further explore unknown options for the benefit of others in the long run. To regulate the human-in-the-loop learning (HILL) game against selfish agents' free-riding, this paper proposes a paradigm shift from centralized to decentralized way of operation that forces agents' local explorations through restricting information sharing. When game theory meets distributed learning, we formulate our decentralized communication mechanism's design as a new multi-agent Markov decision process (MA-MDP), and derive its analytical condition to outperform today's centralized operation. As the optimal decentralized communication mechanism in MA-MDP is NP-hard to solve, we present an asymptotically optimal algorithm with linear complexity to determine the mechanism's timing of intermittent information sharing. Then we turn to non-myopic agents who may revert to even over-explore, and adapt our mechanism design to work. Simulation experiments using real-world dataset demonstrate the effectiveness of our decentralized mechanisms for various scenarios. 
\end{abstract}

\ccsdesc[500]{Networks~Network economics}
\ccsdesc[500]{Computing methodologies~Multi-agent systems}

\keywords{Distributed Multi-armed bandits, Incentivizing Exploration, Dynamic Bayesian game, Communication restriction}
\maketitle

\section{Introduction}
Over the past decade, human-in-the-loop learning (HILL), which involves human agents as both information producers and consumers, has been widely implemented in real-world applications (e.g., TripAdvisor and Waze). \cite{li2024human, kremer2014implementing, dubey2020cooperative, immorlica2019bayesian}). It empowers a human agent to crowdsource other agents' local resources and knowledge for learning large-scale systems (e.g., restaurant ratings, route recommendation, sales promotion). In HILL platforms, agents dynamically choose between exploring unknown options and exploiting high-reward ones based on the platforms' shared information over time \cite{avner2016multi}. The uncertain reward from each service option, along with interaction among agents, poses principal challenges in this dynamic learning problem, necessitating strategic option choices and regulation to optimize social rewards over time. 

All these HILL platforms operate in a centralized way to collect agents' latest observations of new options (e.g., restaurants, routes) and share such information with all in real time, aiming to connect all agents for maximum information exchanges. However, after hearing the central platforms' live updates, human agents are found selfish and unwilling to further explore unknown options for the benefit of others in the long run. For example, selfish agents prefer to exploit established, highly-rated options, quickly converging to popular choices and neglecting many potentially better alternatives \cite{park2008effects,kremer2014implementing,burtch2018stimulating}. Facilitated by the central platforms, they appreciate reviewing the feedback shared by their peers but remain hesitant to explore additional options that could potentially benefit the community \cite{li2023congestion}. This phenomenon suggests that centralized information sharing to connect all in the current practice leads to severe under-exploration and poor social reward in the long run.

The multi-armed bandit (MAB) problem is a representative framework for studying the dynamic exploration-exploitation trade-off. Existing works on the extension to multi-agent MAB problems largely assume that agents are cooperative in communicating with a shared goal of maximizing long-term social reward. They focus on algorithms that enable and coordinate collaborative learning under communication limits or in the absence of centralized communication (e.g., \cite{martinez2019decentralized,dubey2020cooperative,bistritz2020cooperative,madhushani2021one, yang2021cooperative, agarwal2022multi, chang2022online, chen2023demand, xu2024decentralized}). For example, Dubey and Abhimanyu \cite{dubey2020cooperative} and Madhushani et al. \cite{madhushani2021one} investigate the challenge of delayed information transmission among agents to converge globally. Agarwal et al. \cite{agarwal2022multi} aim to reduce communication bits while achieving the target regret bound.

All these works assume human agents will adhere to prescribed algorithms without considering possible deviations for selfish interests. As such, they perceive communication opportunities in centralized information sharing as beneficial, aiming to maximally overcome the constraints of information communication for optimal multi-agent learning results. However, if an agent is selfish and focuses on maximizing her own rewards, receiving increased communication or shared information from others' exploration just discourages her from exploring new options herself. If this free-riding behavior becomes widespread among all agents, the system lacks the necessary exploration to recommend the most beneficial options for all agents in the long run. 

To regulate information sharing and incentivize exploration, one may think of other pricing or auction mechanisms to regulate (e.g., \cite{wang2018multi}\cite{braverman2019multi}), yet these methods require monetary exchange or side-payment, which is hard to implement by information sharing platforms such as Tripadvisor and Waze in practice. Alternatively, several recent studies explore the problem under the concept of Bayesian persuasion and Bayesian incentive-compatible mechanisms (e.g., \cite{kremer2014implementing,tavafoghi2017informational,immorlica2018incentivizing, che2018recommender,immorlica2019bayesian, mansour2020bayesian, sellke2021price,mansour2022bayesian, hu2022incentivizing, simchowitz2023exploration}). For example, Kremer et al. \cite{kremer2014implementing} initiate the research, introducing a recommendation mechanism for incentivizing exploration of options with deterministic rewards by a single, non-repeating agent per round. Then Mansour et al. \cite{mansour2020bayesian} develop a method that makes any bandit algorithm Bayesian incentive-compatible, minimizing regret with stochastic rewards. Immorlica et al. \cite{immorlica2018incentivizing} use a selective data disclosure method to incentivize exploration in the bandit problem, and investigate incentivizing exploration among heterogeneous agents in \cite{immorlica2019bayesian}, where the reward function depends on the agent type and system state. Additionally, Tavafoghi and Teneketzis \cite{tavafoghi2017informational} apply informational incentives to transportation networks to mitigate congestion, where the route conditions are unknown to the agents. However, these methods largely assume myopic agents who only participate in a single round without having long-term reward objectives. Non-myopic human agents with enough sophistication, however, can strategically defeat these methods by misreporting information to fool others and deviating from the prescribed algorithms. We will show that the regulatory challenge evolves from under- to over-exploration.


Our key novelty and main contributions are summarized as follows.

\begin{itemize}
\item \emph{Novel HILL game through decentralized communication mechanism}: To our best knowledge, we are the first to propose the paradigm shift from centralized to decentralized information-sharing to enable the distributed exploration of self-interested agents. Previous HILL studies including distributed MAB works all prefer centralized system for cooperative agents' maximum communication (e.g., \cite{martinez2019decentralized}\cite{madhushani2021one}\cite{chang2022online}), while our work pioneers the use of communication restriction in decentralized system to force non-cooperative agents' local explorations, essentially involving the idea of concealing information. When game theory meets distributed learning, we formulate the dynamic HILL problem as a new multi-agent Markov decision process (MA-MDP) in Section \ref{S2}.

\item\emph{Asymptotically optimal decentralized communication mechanism for myopic agents}: The MA-MDP problem of decentralized communication mechanism design for social welfare optimization is NP-hard, by including integer timing decisions to disable communication among agents and a non-convex objective. Despite this, in Section \ref{S3}, we first analytically derive a sufficient and necessary condition in closed form for the platform to deviate from the prevailing practice of centralized communication, which often converges fast to inadequate exploration. Then we prove that it is best for the platform to intermittently open multiple communication windows among agents over time. Finally, to solve the NP-hard problem, we propose an asymptotically optimal algorithm with linear complexity.

 \item \emph{Optimal decentralized communication mechanism to regulate non-myopic agents}: In Section \ref{S4}, we further address the challenge posed by some sophisticated human agents who non-myopically withhold information, surprisingly causing over-exploration for personal information gain over time. We show that they prefer free-riding and are not willing to share their learned information truthfully until the last time stage. In the dynamic Bayesian game, we successfully design the optimal decentralized communication mechanism to enable one-time communication at a strategic interval which mitigates over-exploration, while maintaining polynomial complexity.

\item \emph{Simulation experiments with real dataset for stochastic rewards and heterogeneous agents}: We conduct extensive simulation experiments together with a real-world dataset on hotel service rating to demonstrate the effectiveness of our decentralized communication design in Section~\ref{S5}. We show our method obviously outperforms the existing centralized communications scheme. Furthermore, we apply our method to two more general cases with stochastic reward per option and heterogeneous agent preferences over the same option. Simulation results demonstrate that our method remains effective in these extended scenarios.
\end{itemize}

\section{System Model and Problem Formulation }\label{S2}
In this section, we first present the general system model for selfish agents' HILL game under any centralized/decentralized communication environment. We then delineate the problem formulation for the human-in-the-loop learning game in the multi-agent Markov decision process (MA-MDP) framework, by maximizing long-term social welfare via our decentralized communication mechanism.
\subsection{System Model on HILL Game}
Consider an information sharing platform with $N$ agents in the set $\mathcal{N}:=\{1,\ldots,N\}$ engaged in the HILL game to choose unknown service options in the set $\mathcal{K}$ over a discrete time horizon $[T]:=\{0,1,\ldots,T\}$. We focus on the challenging case of a large-scale unknown system with far more options than agents. The reward from choosing option $k\in\mathcal{K}$ is $r_k$, an unknown and deterministic value independently drawn from a general cumulative distribution function $F(\cdot)$ in a normalized continuous support [0, 1] with mean reward $\mu$.\footnote{We can also extend our analysis to allow that options follow a finite number of \( L \) different distributions \( F_1(\cdot), \ldots, F_L(\cdot) \) with means \( \mu_1 \geq \ldots \geq \mu_L \). There, agents will first explore options from \( F_1(\cdot) \) class until exhausted, then proceed to \( F_2(\cdot) \), and so on. We can similarly derive the exploration thresholds for myopic agents in Lemma \ref{myopic_S} (by replacing $\mu$ by the mean reward of the best remaining option class) and for non-myopic agents in Lemma \ref{NonMyFull}, and our decentralized communication mechanisms still work.} At each time slot $t \in [T]$, Fig. \ref{f1} shows that each agent $n\in\mathcal{N}$ selects an option $a_{n}(t)\in\mathcal{K}$ and observes the reward $r_{a_{n}(t)}$, which may be centrally recorded and shared to other agents for their future consideration by the platform. Thereafter, this reward value $r_{a_n(t)}$ remains constant over the finite time horizon $[T]$. Such deterministic reward models are common in the HILL literature for many applications of static service qualities (e.g., hotels, restaurants, and travel routes) in finite time (e.g., \cite{kremer2014implementing,tavafoghi2017informational, immorlica2019bayesian,david2014infinitely,david2015refined}). For the case of stochastic reward, analyzing the optimal exploration of stochastic rewards within a finite time horizon or for many options is statistically challenging, and there is currently no optimal solution even for centralized systems \cite{lattimore2020bandit}. Thus, we will conduct experiments using a real dataset in Section \ref{S5} to demonstrate that our mechanisms remain effective as long as the reward variation over time is limited.

The information sharing platform, acting as the social planner, regulates whether to allow agents' mutual communications in each time slot. This is equivalent to determining whether to publicize the agents' shared information in real time on the platform. If no communication is allowed, agents need to locally decide the exploration-exploitation tradeoff in a decentralized way. Define any decentralized communication mechanism to include $M\geq1$ no-communication time periods over time as set $ \mathcal{M}=\cup_{m=1}^M\{T_m,\ldots,\allowbreak T_m+\Delta_m-1\}$ where communication is not allowed for $t\in \mathcal{M}$. Each no-communication period \( m \in \{1, \ldots, M\} \) starts at \( T_m \) and lasts \(\Delta_m \) consecutive time slots, ending at \( T_m + \Delta_m - 1 \). Two nearest periods are separated by at least one communication slot, i.e., \( T_{m+1} > T_m + \Delta_m \). $\mathcal{M}=\emptyset$ tells that the platform does not disable communication among agents and makes their shared information open-access all the time. $\mathcal{M}$ is established by the social planner prior to the initiation of the learning process and is announced to all. Fig. \ref{f1} illustrates the agent communication under the mechanism $\mathcal{M}$ in the information sharing platform. With open access to centrally shared information, agents decide to explore or exploit by considering both shared and personal information. When communication is restricted, agents rely on their own information to act.
\begin{figure}
    \centering
    \includegraphics[width=0.45\textwidth]{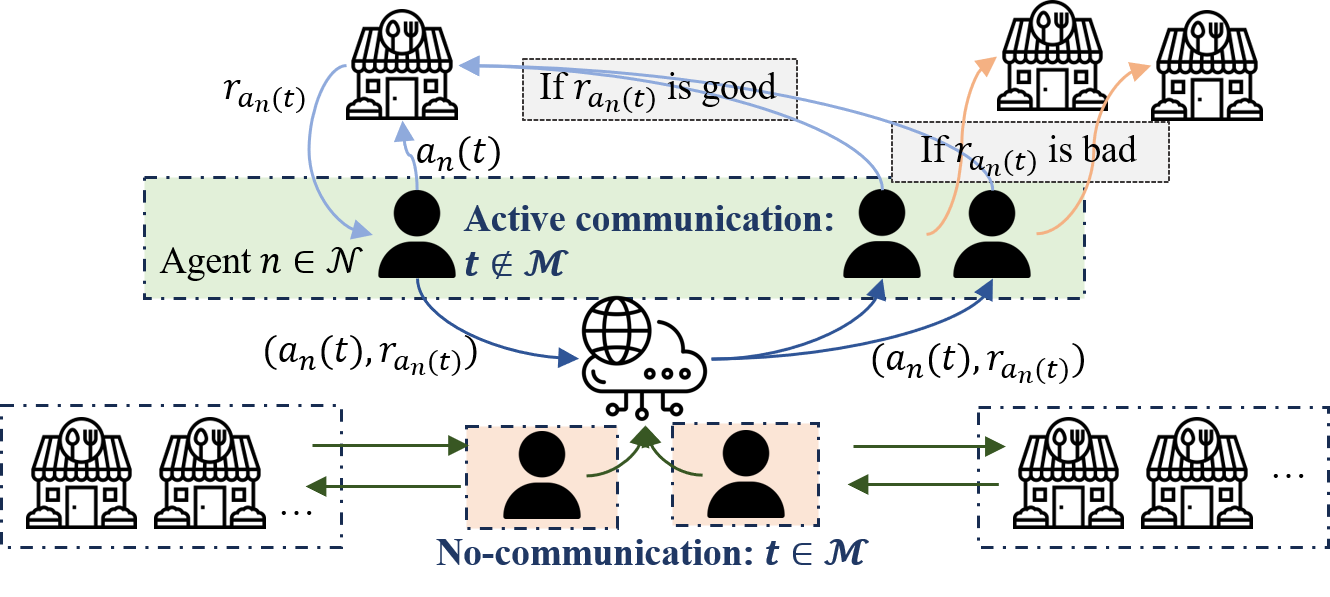}
    \caption{Decentralized communication mechanism $\mathcal{M}$ in the information sharing platform, where each agent $n\in \mathcal{N}$ can exchange his option observation ($a_n(t), r_{a_n(t)}$) with others for time $t\notin \mathcal{M}$, and cannot share for $t\in \mathcal{M}$. Given sharing (if any) at time $t$, other agents at the next time may exploit the same option $a_n(t)$ if its reward $r_{a_n(t)}$ is good or proceed with some other options.}
    \label{f1}
\end{figure}

We next explain the decision-making of each agent at time $t$. Denote the set of options whose reward information is known by agent $n$ by time slot $t$ as $I_n(t)$, and the set of corresponding rewards as $R_n(t) = \{(k, r_{k}) \,|\, k \in I_n(t)\}$. At the beginning of time slot $t$, each agent $n$ independently decides $a_n(t)$. If she chooses to explores a new option $a_n(t)\notin I_n(t)$, she updates her option-reward information by $I_n(t+1) = I_n(t) \cup \{a_n(t)\}$ and $R_n(t+1) = R_n(t) \cup \{(a_n(t), r_{a_n(t)})\}$. If she opts to exploit the option $a_n(t)\in I_n(t)$, her information remains unchanged. If communication is permitted at time slot with $t\in [T]\setminus\mathcal{M}$ in the upper-part of Fig. \ref{f1}, each agent decide whether to share truthful information $I_n(t+1)$ and $R_n(t+1)$ with all other agents, and updates her information again as $I_n(t+1) = \cup_{i\in\mathcal{N}_t\cup\{n\}}I_i(t+1)$ and $R_n(t+1) = \cup_{i\in\mathcal{N}_t\cup\{n\}}R_i(t+1)$, where $\mathcal{N}_t$ is the subset of agents who share information at time $t$.

\subsection{Problem formulation as MA-MDP}
Define the strategy of agent $n\in\mathcal{N}$ at time $t$ as $s_n(t|R_n(t), \mathcal{M})$ to determine the option $a_n(t)$. If myopic, agent $n$'s strategy $s_n$ aims to maximize her immediate reward in each specific time step $t$ (as in \cite{kremer2014implementing}\cite{mansour2022bayesian}), that is,
\begin{align}
    \max_{s_n}{E}[r_{s_n(t)}|R_n(t),\mathcal{M}].\label{Myoobj}
\end{align}
Her exploration results affect the others' through possible communications in $[T]\setminus\mathcal{M}$, and the expected reward should be jointly determined by other agents' exploration results. 

If non-myopic, agent $n$'s strategy $s_n$ at current time $t$ aims to maximize the sum of expected rewards from the current time step to the end, formulated as:
\begin{align}
\max_{s_n}\sum^{T}_{\tau=t}{E}[r_{s_n(\tau)}|R_n(\tau),\mathcal{M}],\label{nonmyoobj}
\end{align}
which implies that non-myopic agents may strategically sacrifice the immediate benefit of the current slot for their own long-term gain. This requires jointly considering her possible exploration outcomes and the future option information shared by others under the communication mechanism. We naturally model this as a dynamic Bayesian game with hidden information evolving over time, where the strategies of agents affect each other and should be interactively determined for the stable outcome of Bayesian equilibrium.

In our HILL problem, the platform or social planner cannot force each agent to adopt a cooperative strategy, but can apply a decentralized communication mechanism $\mathcal{M}$ in the platform. The social planner's goal is to maximize the total time reward among all $N$ agents as the social welfare by properly designing $\mathcal{M}$ and predicting agents' best response strategies to each other under $\mathcal{M}$ over time \footnote{One may further apply a discount factor to future rewards for non-myopic agents in (\ref{nonmyoobj}) or the social welfare in (\ref{socialobj}), and our following analysis still works. Intuitively, a shorter time horizon $T$ of our current model implies a larger time discount.}:
\begin{align}
\max_{\mathcal{M}}&\;{V_{\mathcal{M}}}=\max_{\mathcal{M}}\sum^{T}_{\tau=0}\sum^{N}_{n=1}{E}[r_{s_n(\tau)}|R_n(\tau),\mathcal{M}],\nonumber\\
s.t. &\;s^*_n=\begin{cases}
\arg\max_{s_n}(\ref{Myoobj}), \text{if } n\in\mathcal{N} \text{ myopic},\\
\arg\max_{s_n}(\ref{nonmyoobj}), \text{if } n\in\mathcal{N} \text{ non-myopic}.
\end{cases}\label{socialobj}
\end{align}

We next formulate our HILL problem for non-myopic agents as a MA-MDP to solve the agents' optimal strategy under any mechanism $\mathcal{M}$, where the problem for myopic agents can be viewed as a special case of such formulation. First, we observe that because the reward is deterministic, agent \( n \) will always choose the best-known option when she decides to exploit. Therefore, in her information set \( R_n(t) \), only the maximum reward and the corresponding option, denoted as $(\bar{a}_{n}(t), m_n(t)) = \bigg(\underset{k \in I_n(t)}{\arg\max}\{r_k\}, \max_{k \in I_n(t)}\{r_k\}\bigg)$ will influence her actions. We next introduce our formulation of MA-MDP as follows:
\begin{itemize}
    \item \textbf{States and environment}: Agent $n$'s state is the maximum explored reward \(m_n(t)\) at time $t$, and \(m_n(0) = 0\) initially. The communication state acts as the environment in this MA-MDP, where \(t \notin \mathcal{M}\) indicates active communication and \(t \in \mathcal{M}\) indicates blocked communication.
    \item \textbf{Action and reward}: Regarding option choice strategy, given $(\bar{a}_n(t), m_n(t))$ at time $t\in [T]$, agent $n$ will either exploit the currently best option $\bar{a}_n(t)$ to receive reward $m_n(t)$ or randomly explore a new option outside $I_n(t)$ to receive expected reward $\mu$. Additionally, if \( t \notin \mathcal{M} \), agent \( n \) needs to decide an information sharing policy to share what option-reward information with all other agents. 
    \item \textbf{State transition probability}:  For agent \( n \) with state \( m_n(t) \), the space of the next state $m_n(t+1)$ is any continuous value in $[m_n(t),1]$. If she exploits $\bar{a}_n(t)$,  $m_n(t+1)$ remains unchanged. Conversely, if she explores a new option, her state of maximum reward transitions to \(m_n(t+1) =\max(r, m_n(t)) \), where \( r \sim F(\cdot) \) represents the random variable of the option reward. If \( t \notin\mathcal{M} \), her state transition further becomes \( \max_{j \in \mathcal{N}_t\cup\{n\}}(m_j(t)) \) under exploitation, and $\max(r,x)$ whe-re $x=\max_{j \in \mathcal{N}_t\cup\{n\}} {(m_j(t))}$ under exploration. Based on this process, we summarize the state transition probability as follows. When $a_n(t)\in I_n(t)$ and $t\in \mathcal{M}$, $m_n(t+1)$ equals $m_n(t)$ with probability 1. Denote $G_t(\cdot)$ as the CDF of any other agent $j$'s shared state $m_j(t)$ to agent $n$ at the end of $t$ if $t\notin\mathcal{M}$. The PDF of the rest state transitions is:
\end{itemize}
    \begin{align}
        &p(m_n(t+1)=m|a_n(t), t, m_n(t))=\nonumber\\
        &\begin{cases}
            &\!\!\!\!\frac{dF(r)G_t(r)^{|\mathcal{N}_t\cup\{n\}|}}{dr}|_{r=m},\; \text{if }a_n(t)\notin I_n(t), t\notin \mathcal{M},\\
            &\!\!\!\!\frac{dF(r)}{dr}|_{r=m}, \qquad\;\;\;\;\;\;\;\;\;\ \text{ if }a_n(t)\notin I_n(t), t\in \mathcal{M},\\
            &\!\!\!\!\frac{dG_t(r)^{|\mathcal{N}_t\cup\{n\}|}}{dr}|_{r=m}, \ \;\;\;\;\text{ if }a_n(t)\in I_n(t), t\notin \mathcal{M}.
        \end{cases}\label{transition}
    \end{align}
For the MA-MDP described above, the Bellman equation for agent $n$'s optimal long-term reward \( U_n(t \mid m_n(t), \mathcal{M}) \) at time \( t \), under the state \( m_n(t) \) and the communication environment \( \mathcal{M} \), is expressed as:
\begin{align}
    &U_n(t|m_n(t), \mathcal{M}) = \nonumber\\
    & \!\!\!\! \begin{cases}
        &\!\!\!\!\max(\mu+E[U_n(t+1|\max(r,\max_{j\in\mathcal{N}_t\cup\{n\}}m_j(t)),\mathcal{M})],\\
    & \ m_n(t)+E[U_n(t+1|\max(\max_{j\in\mathcal{N}_t\cup\{n\}}m_j(t)),\mathcal{M})]),\\
    &\qquad\qquad\qquad\qquad\qquad\qquad\quad\;\;\;\;\;\;\;\;\text{ if } t\notin \mathcal{M},\\
    &\!\!\!\!\max(\mu+E[U_n(t+1|\max(r, m_n(t)),\mathcal{M})],\\
    & \ m_n(t)+E[U_n(t+1|m_n(t),\mathcal{M})]),\ \text{ if } t\in \mathcal{M},\label{bell}
    \end{cases}
\end{align}
which also depends on other agents' states and strategies, influenced by her information-sharing strategy.



By backward induction in this multi-agent MDP, we can identify that the optimal option choice strategy $s_n(t|R_n(t),\mathcal{M})$ must be a threshold-based strategy to decide to explore or exploit, as demonstrated in the following lemma: 
\begin{lemma}\label{thresholdstrategy}
Under any communication mechanism $\mathcal{M}$, each agent $n$'s optimal option choice strategy $s_n(t| m_n(t),\mathcal{M})$ at time $t\!\in \![T]$ follows an exploration threshold sequence $\{u_t\}_{t=1}^{T}$ over time. That is, she will explore a new option if $m_n(t)<u_t$ and otherwise stick with the currently best option $\bar{a}_n(t)$.
\end{lemma}

In this MA-MDP, each agent has two actions: exploration/exploi-tation and information sharing to decide. Solving \(\{u_t\}_{t=1}^{T}\) for non-myopic agents is challenging because their exploration thresholds are influenced by the potential information exchanged from others in any future $t\notin\mathcal{M}$. Specifically, they tend to reduce exploration efforts when anticipating useful reward messages from others in the near future. This requires consideration of these two-dimensional actions to jointly determine all agents' exploration thresholds \(\{u_t\}_{t=1}^{T}\) and information sharing strategy in the equilibrium of the dynamic Bayesian game under the communication mechanism $\mathcal{M}$.

Compared to non-myopic agents, myopic agents have a straightforward strategy, with a constant threshold \( u_t = \mu \), as they only care about the current time slot's reward. However, the dynamic exploration process of myopic agents under various possible communication mechanisms remains complicated to analyze. We will show later that designing the optimal decentralized communication mechanism for myopic agents is already an NP-hard problem to solve.

\section{Decentralized Communication Mechanism for Myopic Agents}\label{S3}
In this section, based on the MA-MDP framework, we will address the questions of whether and when decentralized communication can be more advantageous than the current practice of centralized communications for myopic agents in the platform. Then we will design a decentralized communication mechanism to maximize social welfare (\ref{socialobj}) by solving an NP-hard problem. 
\subsection{Analysis of Existing Centralized Communication Policy}
We will first examine the centralized communications policy to give all agents open access to their shared information, where $\mathcal{M}=\emptyset$. We begin by presenting the optimal learning strategy of each myopic agent $n$ for maximizing her immediate reward in (\ref{Myoobj}) in the dynamic Bayesian game. We solve the exploration threshold $u_t$ in Lemma \ref{thresholdstrategy}, the weakly dominant information sharing strategy, and social welfare below.

\begin{lemma}\label{myopic_S}
    At each time $t\in[T]$, each myopic agent $n$'s option exploration threshold sequence is $u_1=\ldots=u_{T}=\mu$. Once permitted to communicate with $t\notin\mathcal{M}$, myopic agents will truthfully share their latest information. Then the long-term social welfare in (\ref{socialobj}) is: \begin{align}
        V_{\emptyset}=&N(T+1)\int^1_{\mu}rd\frac{F(r)^N-F(\mu)^{N}}{1-F(\mu)^{N}}\nonumber\\
        &-N\frac{1-F(\mu)^{N(T+1)}}{1-F(\mu)^N}\int^1_{\mu}\frac{1-F(r)^N}{1-F(\mu)^{N}}dr,
        \label{fullM_socialwelfare}
    \end{align}
    where agents quickly stop exploring any new option with expected exploration time:  
    \begin{align}
        (1-F(\mu)^{N(T+1)})/(1-F(\mu)^{N})\label{M_exNum}.
    \end{align}
\end{lemma}

The proof of Lemma \ref{myopic_S} is given in Appendix \ref{appendix1}. Myopic agents' centralized communications enable the community to converge quickly once some agents find an option with reward larger than threshold $\mu$. This fast convergence in (\ref{M_exNum}) implies a lack of essential exploration with poor performance in (\ref{fullM_socialwelfare}), particularly for a large number of options available. Additionally, given that agents' communication takes place after their options are chosen, agent $n$'s misreporting of her latest information $R_n(t)$ does not alter her immediate reward in (\ref{Myoobj}). Thus, she is willing to truthfully share the information. Later in Section \ref{S4}, we will show that this is not the case for non-myopic users who may withhold information to gain in the long run.


To identify the sufficient and necessary conditions under which centralized communications should be replaced by decentralized communication, we first meticulously choose a special decentralized communication pattern with a single no-communication time window $\bar{\mathcal{M}}=\{T_1,\ldots,T_1+\Delta_1-1\}$ to compare with centralized communications. We derive the long-term social welfare objective under $\bar{\mathcal{M}}$ as:
\begin{align}
    V_{\bar{\mathcal{M}}}=NF(\mu)^{NT_1}\bigg((T-T_1-\Delta_1)y_{\Delta_1+1}-\sum_{i=1}^{\Delta_1}x_i\bigg)+V_{\emptyset},\label{specialM}
\end{align}
where 
\begin{align}
x_i=&\int^1_{\mu}\bigg(\frac{F(r)-F(\mu)}{1-F(\mu)}+F(\mu)^i\frac{1-F(r)}{1-F(\mu)}\bigg)\nonumber\\
        &-\bigg(\frac{F(r)^N-F(\mu)^N}{1-F(\mu)^N}+F(\mu)^{iN}\frac{1-F(r)^N}{1-F(\mu)^N}\bigg)dr,\label{xi}\\
y_{i}=&\int^1_{\mu}\bigg(\frac{F(r)^N-F(\mu)^N}{1-F(\mu)^N}+F(\mu)^{iN}\frac{1-F(r)^N}{1-F(\mu)^N}\bigg)\nonumber\\
        &-\bigg(\frac{F(r)-F(\mu)}{1-F(\mu)}+F(\mu)^{i}\frac{1-F(r)}{1-F(\mu)}\bigg)^Ndr,\label{yi}
\end{align}
for $i\in\{1,\ldots,\Delta_1+1\}$, and $V_{\emptyset}$ is given in (\ref{fullM_socialwelfare}). Here $x_i$ can be interpreted as the reward loss due to no-communication among agents to share and exploit the currently best options, while $y_i$ tells the benefit of boosting self-exploration using no-communication for the future $T-T_1-\Delta_1$ time slots. 

Equation (\ref{specialM}) suggests that if its bracket $(T-T_1-\Delta_1)y_{\Delta_1+1}-\sum_{i=1}^{\Delta_1}x_i>0$, decentralized communication $\bar{\mathcal{M}}$ outperforms centralized communications, which provide us a sufficient condition for decentralized communication mechanism $\bar{\mathcal{M}}$ to outperform. Intriguingly, we find that it is also necessary for the existence of decentralized communication mechanisms to outperform centralized communication. This can be done by showing that our later derivation of social welfare in (\ref{VM}) under general decentralized communication pattern $\mathcal{M}=\cup_{m=1}^M\{T_m,\ldots,T_m+\Delta_m-1\}$ degenerates to the particular $\bar{\mathcal{M}}$ in the comparison with the centralized communications' $V_{\emptyset}$ in (\ref{fullM_socialwelfare}). 

\begin{proposition}\label{p1}
    Both the sufficient and necessary condition for our decentralized communication mechanism to outperform the centralized communications policy is given below in closed-form:
    \begin{align}
        T>\min_{\Delta\in\{1,\ldots, T-1\}}\left\{\frac{\sum^{\Delta}_{i=1}x_i}{y_{\Delta+1}}+\Delta\right\} .\label{eq_lemmq3.2}
    \end{align}
\end{proposition}

The proof of Lemma \ref{p1} is given in Appendix \ref{Appendix11}. (\ref{eq_lemmq3.2}) indicates that as long as time horizon $T$ is large or the reward distribution $F(\cdot)$ in (\ref{xi}) and (\ref{yi}) becomes positively skewed with smaller mean $\mu$, it is better for the platform to replace the prevailing centralized communications policy with decentralized communication to encourage each agent to self-explore and find an option with reward larger than $\mu$. Once communication is allowed later at $T_1+\Delta_1$, the community picks up the best option of all agents with a reward much better than $\mu$. 

\subsection{Decentralized Communication Mechanism}

Given (\ref{eq_lemmq3.2}), we investigate how to design the best decentralized communication mechanism $\mathcal{M}^*_1=\cup_{m=1}^M\{T_m^*,\ldots,T_m^*+\Delta_m^*-1\}$ for myopic agents to maximize the social welfare objective in (\ref{socialobj}). Yet the further derivation for equation (\ref{socialobj}) under any \(\mathcal{M}\) is complicated because agents, without frequent communication, do not converge synchronously and exhibit various case combinations by observing different option rewards in each no-communication time window. Through derivation using $x_i$ in (\ref{xi}) and $y_i$ in (\ref{yi}), we have the following.

\begin{lemma}\label{L333}
Given any decentralized communication mechanism $\mathcal{M}=\cup_{m=1}^M\{T_m,\ldots,T_m+\Delta_m-1\}$, the expected social welfare to consider all time slots is given by:
\begin{align}
    V_\mathcal{M}=&N\sum^M_{m=1}F(\mu)^{NT_m}\bigg((T-T_m-\Delta_m)y_{\Delta_m+1}-\sum_{i=1}^{\Delta_m}x_i\bigg)+V_{\emptyset},\label{VM}
\end{align}
where $V_{\emptyset}$ is given in (\ref{fullM_socialwelfare}). The expected number of exploration is increased from (\ref{M_exNum}) under centralized communications to
\begin{align}
\sum^{T_1}_{t=0}F(\mu)^{Nt}+\sum^M_{m=1}\bigg(\sum^{T_{m+1}}_{t=T_m+\Delta_m+1}\!\!\!\!\!\!F(\mu)^{Nt}+\sum^{\Delta_m}_{t=1}F(\mu)^{NT_m+t}\bigg).\label{M_exNum_R}
\end{align}
\end{lemma}

The proof of Lemma \ref{L333} is given in Appendix \ref{Appendix16}. In (\ref{VM}), when we extend the no-communication duration from \(\Delta_m\) to \(\Delta_m+1\) in mechanism $\mathcal{M}$, we incur immediate exploitation loss of \(x_{\Delta_m+1}\), yet the future exploration benefit increases from \(y_{\Delta_m}\) to \(y_{\Delta_m+1}\). We need to balance them well to decide the optimal timing of no-communication.

Given (\ref{VM}), we can formulate the optimization problem to determine the optimal decentralized communication mechanism \(\mathcal{M}^*_1\) with \(T_m\) and \(\Delta_m\). This problem is NP-hard due to the inclusion of integer timing decisions for disabling communication among agents and a non-convex objective. However, equation (\ref{VM}) decreases with \(T_m\) and this feature helps simplify the optimization problem, as detailed in the subsequent lemma:
\begin{lemma}\label{34}
The optimal decentralized communication mechanism \(\mathcal{M}^*_1\)=\(\cup_{m=1}^M\{T_m^*,\ldots,\) \(T_m^*+\Delta_m^*-1\}\) must satisfy $T^*_1=0,T^*_{m+1}=T^*_m+\Delta^*_m+1$, for all $m\in\{1,\ldots,M-1\}$.
\end{lemma}

In the optimal mechanism, we need to disable communication initially with $T_1^*=0$ to encourage agents' own exploration. Additionally, we only include a single communication round between two neighboring no-communication windows, as it suffices for agents to share their option information. Then the problem of jointly deciding $T_m^*$ and $\Delta_m^*$ is simplified to deciding $\Delta_m^*$ only, and the optimization problem can be simplified as:
\begin{align}
&\max_{M,\{\Delta_m\}_{m=1}^M}N\bigg((T-\Delta_1)y_{\Delta_1+1}-\sum_{i=1}^{\Delta_1}x_i\bigg)\nonumber\\
&+\sum^M_{m=2}NF(\mu)^{N\sum_{i=1}^{m-1}(\Delta_i+1)}\bigg((T-\sum_{i=1}^m\Delta_i-m+1)y_{\Delta_m+1}-\sum_{i=1}^{\Delta_m}x_i\bigg)\nonumber\\
    &\quad\quad\quad\text{s.t. }\;\Delta_i\in {Z}^+\;,\forall i\in\{1,\ldots,M\},\;\sum_{i=1}^M\Delta_i\leq T-M,\nonumber\\
    &\;\;\quad\quad\quad\quad\; M\in\{1,\ldots,\lfloor T/2 \rfloor\}.
    \label{OPT}
\end{align}

However, this problem is still a difficult non-convex problem to solve with combinatorial decisions in $\{\Delta_1,\ldots,\Delta_M\}$ and $M$. While a recursive algorithm can be developed for it, allowing for some pruning of the search space based on condition (\ref{eq_lemmq3.2}), this algorithm still runs in exponential time and it is hard to optimally solve for non-small horizon $T$. 

To tell what the optimal mechanism looks like, we use an illustrative example with $T=30$ for affordable computation. We opt for the Beta distribution as the prior $F(\cdot)$ and run two experiments with mean option rewards of \(0.024\) and \(0.33\). Their corresponding optimal mechanisms are depicted in Fig. \ref{fCRM}. The no-communication window duration decreases over time for \(\mu=0.024\), as less time remains, reducing the need for long communication windows to encourage exploration. Additionally, as \(\mu\) increases from 0.024 to 0.33, both the number \(M\) of no-communication windows and their durations \(\Delta_m\) decrease. As the mean reward \(\mu\) increases, myopic agents have a higher exploration threshold, and the system with centralized communications incurs less under-exploration. Thus, the optimal mechanism requires fewer no-communication windows.

\begin{figure}
    \centering
    \includegraphics[width=0.5\textwidth]{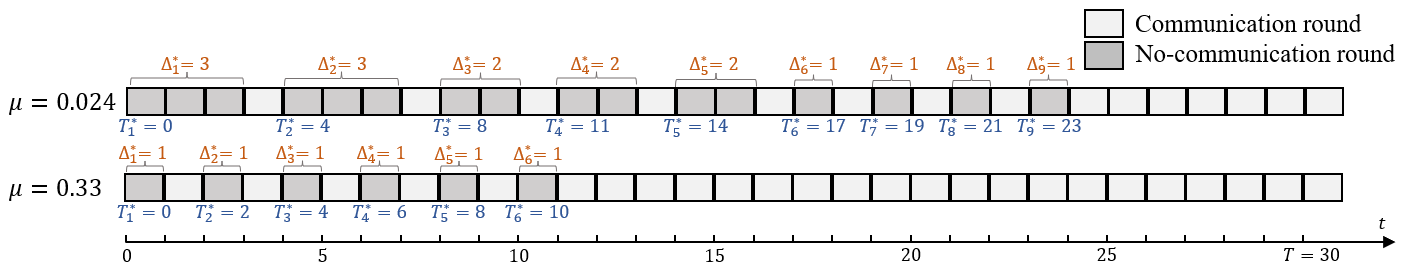}
    \caption{Optimal decentralized communication mechanism for 2 different Beta distributions with means 0.024 and 0.33, respectively.}
    \label{fCRM}
\end{figure}

For the tractability of solving problem (\ref{OPT}), we reduce it to an approximate version below that concentrates on a single no-communication window $\tilde{\mathcal{M}}_1=\{0,\ldots,\Delta_1-1\}$, inspired by observing (\ref{OPT}) that the benefit from the $m$-th no-communication window is discounted by $F(\mu)^{N\sum^{m-1}_{i=1}(\Delta_i+1)}$, exponentially diminishing over $m$. By setting $T_1^*=0$ for (\ref{specialM}), the approximate version of problem (\ref{OPT}) is:
\begin{align}
    \max_{\Delta_1\in\{1,\ldots,T-1\}}N\bigg((T-\Delta_1)y_{\Delta_1+1}-\sum_{i=1}^{\Delta_1}x_i\bigg).\label{pro2}
\end{align}
Consequently, we propose Algorithm \ref{A1} to solve this approximate problem efficiently. At the beginning of the search, the algorithm first checks the assumption (\ref{eq_lemmq3.2}) on the total time horizon \(T\) to disable communication. If (\ref{eq_lemmq3.2}) is satisfied, it iterates over all \(\Delta_1 \in \{1,\ldots,T-1\}\) and returns the one $\Delta_1^*$ for which mechanism \(\tilde{\mathcal{M}}_1^* = \{0, \ldots, \Delta_1^* - 1\}\) yields the maximum social welfare \(V_{\tilde{\mathcal{M}}_1^*}\).

\begin{algorithm}
\caption{Asymptotically Optimal Decentralized Communication Mechanism $\tilde{\mathcal{M}}_1^*$}\label{A1}
\begin{algorithmic}[1]
\REQUIRE $T$, $N$, $F(\cdot)$
\IF{$T \leq \min_{\Delta_1 \in \{1, \ldots, T-1\}} \left\{\frac{\sum^{\Delta_1}_{i=1} x_i}{y_{\Delta_1+1}} + \Delta_1\right\}$ in (\ref{eq_lemmq3.2})}
    \RETURN $\Delta_1^*=0$ for centralized communications
\ENDIF
\FOR{$\Delta_1 = 1$ to $T-1$}
    \IF{$V_{\tilde{\mathcal{M}}_1^*}$ in (\ref{pro2}) is greater than $\bar{V}$}
        \STATE $\Delta_1^* = \Delta_1$, $\tilde{\mathcal{M}}_1^* = \{0, \ldots, \Delta_1^*-1\}$, $\bar{V} = V_{\tilde{\mathcal{M}}_1^*}$
    \ENDIF
\ENDFOR
\RETURN $\tilde{\mathcal{M}}^*_1 = \{0, \ldots, \Delta_1^*-1\}$
\end{algorithmic}
\end{algorithm}
\begin{theorem}\label{T3.6}
       Algorithm \ref{A1} operates with a linear complexity of $O(T)$ and guarantees that the social welfare, $V_{\tilde{\mathcal{M}}_1}$, achieved by the approximate mechanism, is a $(1-\frac{F(\mu)^{2N}-F(\mu)^{TN}}{1-F(\mu)^{TN}})$-approximation of the social optimum $V_{\mathcal{M}^*_1}$ for problem (\ref{OPT}). Algorithm \ref{A1} is asymptotically optimal as $N\rightarrow\infty$, and leads to a linear-order social welfare gain from (\ref{fullM_socialwelfare}) in centralized communications over time.
\end{theorem}

The proof of Theorem \ref{T3.6} is given in Appendix \ref{Appendix15}. Note that even as \(N \rightarrow \infty\), if \(T\) is sufficiently large, the centralized communications policy remains non-optimal by (\ref{eq_lemmq3.2}). 

We now run simulations to examine Algorithm \ref{A1} besides the theoretical results in Theorem \ref{T3.6}. We run Algorithm \ref{A1} across various Beta distributions for option distribution with different parameter settings of mean reward value $\mu$ of $0.3$, $0.5$ and $0.6$. 

We examine the performance gain $V_{\tilde{\mathcal{M}}_1^*}-V_\emptyset$ of Algorithm \ref{A1} from centralized communications versus $T$ in Fig. \ref{f2.31}, and it also shows the upper bound of optimal performance gain $V_{{\mathcal{M}}_1^*}-V_\emptyset$ since the optimum is difficult to solve for large $T$ in hundreds. Fig. \ref{f2.31} shows that our performance gain $V_{\tilde{\mathcal{M}}_1^*}-V_\emptyset$ scales up linearly with \(T\), in accordance with Theorem \ref{T3.6}. Additionally, even for finite $N=100$ as used by Fig. \ref{f2.31}, the social welfare achieved by our asymptotically optimal mechanism is almost the same as the upper bound of the optimum. Notably, our mechanism's performance gain increases as \(\mu\) decreases, since the centralized communications benchmark encourages agents to converge more quickly to lower exploration threshold $\mu$.
\begin{figure}
    \centering
\includegraphics[width=0.3\textwidth]{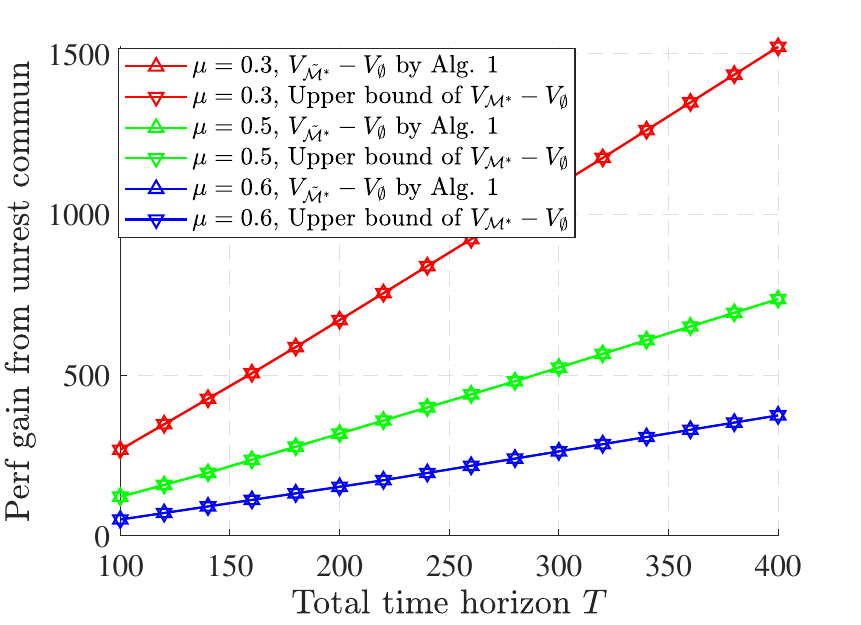}
    \caption{Performance gains of Algorithm \ref{A1} and upper bound of optimum from centralized communications, versus $T$ and $\mu$ of 0.3, 0.5 and 0.6.}
        \label{f2.31}
\end{figure}

\section{Dcentralized Communication Mechanism for Non-myopic Agents}\label{S4}
In this section, we extend to consider non-myopic human agents who may strategically sacrifice the immediate benefit of the current time slot for her own long-term reward in (\ref{nonmyoobj}). After analyzing non-myopic agents' strategy under MA-MDP game, we develop a new decentralized communication mechanism strategically tailored for non-myopic agents.

\subsection{Analysis of Existing Centralized Communications Policy}
We first examine non-myopic agents' information communication strategy under centralized communications of the dynamic Bayesian game, which is different from that of myopic agents in Lemma \ref{myopic_S}. Non-myopic agents, to maximize long-term rewards in (\ref{nonmyoobj}), consider all potential outcomes in their strategic choices in the MA-MDP game of Fig. \ref{f1} and may not truthfully share information with others. The imperfect and incomplete information in this game arises from agents not knowing how many will share at each time $t$ and being unaware of other agents' $m_n(t)$. However, based on the following weakly dominant strategy for information sharing, the Bayesian perfect equilibrium in the dynamic Bayesian game can be identified.

\begin{lemma}\label{sharing}
    In each time slot $t\in[T]$ with centralized communications, it is a perfect Bayesian equilibrium for each agent $n\in\mathcal{N}$ to withhold any useful information until time slot \( T-1 \), and to share truthful information in \( T-1 \) with any probability.
\end{lemma}

If agent \( i \) believes that another agent \( n \) is willing to share their maximum reward information at time \( t \) with a positive probability \( q \), the best response for agent \( i \) is to withhold sharing before $t$ with agent \( n \) to maximize the expected reward information received from \( n \), as determined by problem (\ref{bell}). If this belief is 0 for all time slots, withholding information remains a weakly dominant strategy for agent \( i \), as sharing would not impact the expected information from \( n \), given \( n \) would not share in return. Consequently, withholding information until the final communication round forms an equilibrium in the information-sharing strategy, as agents no longer have the opportunity to benefit from others' exploration.  

At time \( T-1 \), since sharing probabilities in \([0,1]\) for all agents are weakly dominant strategies, infinitely many perfect Bayesian equilibria exist, corresponding to various combinations of sharing probability for each agent. However, the equilibrium where every agent shares with a probability of 1 at \( T-1 \) is Pareto optimal. Thus, we focus on this specific perfect Bayesian equilibrium for further analysis, where $\mathcal{N}_{T-1}=\mathcal{N}$ and $\mathcal{N}_t=\emptyset$ for $t<T-1$ in (\ref{transition}) and (\ref{bell}).

At the Pareto optimum of the equilibrium, agents will expect to use other's information at time $T-1$ to improve their final reward at the last time slot $T$ based on the reward function (\ref{bell}). Before time slot $T-1$, agent 
$n$ forms a belief $G_{T-1}(\cdot)$ regarding the maximum reward $m_i(T-1)$ shared by any other agent $i \neq n$ at $T-1$, which should match with her own distribution of $m_n(T-1)$ in the Bayesian equilibrium. Besides the information communication strategy, they will further decide their option choice strategy based on this belief before time $T-1$ and (\ref{bell}), following the exploration threshold sequence $\{u_t\}_{t=1}^T$ in Lemma \ref{thresholdstrategy}.

\begin{lemma}\label{NonMyFull}
Under centralized communications, non-myopic agents' exploration threshold sequence $\{u_t\}_{t=1}^T$ is strictly decreasing with $u_1>\ldots>u_T=\mu$, which are unique solutions to:
\begin{align}
&u_t-\mu= \begin{cases}
    &\! \! \! \! \! (T-t-1)(\int^1_{u_t} 1-F(r)dr)+\\
   &\! \! \! \! \! \int^1_{u_t}G_{T-1}(r)^{N-1}(1-F(r))dr,\;\text{if }t=1,\ldots, T-1,\\
    &\! \! \! \! \! 0,\quad\quad\quad\quad\quad\quad\quad\quad\quad\quad\;\;\;\;\;\;\text{if }t=T,\label{utfull}\end{cases}
\end{align}
where $G_{T-1}(r)$ represents the distribution of each agent $n$'s maximum reward $m_n(T-1)$ by time $T-1$ determined by $\{u_t\}_{t=1}^{T-1}$ in the following:
\begin{align}
G_{T-1}(r)=\left\{  
        \begin{array}{l@{\;}l}
        F(r)^t-(1-F(r))\sum^{T-1}_{i=t}F(u_i)^i,\;&\text{if }u_t<r\leq u_{t-1}\\
        &\;\;\text{ with }u_0=1, \\
         &\;\;\;\text{and }1\leq t\leq T-1,\\
        F(r)^{T},\;&\text{if }r\leq u_{T-1}.
        \end{array}
\right.\label{belief}
\end{align}
Different from myopic agents' insufficient exploration time in (\ref{M_exNum}), non-myopic agents will over-explore with a longer expected number of explorations:
\begin{align}
    1+F(u_T)^{TN}+\sum_{t=1}^{T-1}F(u_t)^t.\label{Non_exNum}
\end{align}
\end{lemma}

The proof of Lemma \ref{NonMyFull} is given in Appendix \ref{appendix4}. The threshold equation (\ref{utfull}) is derived from backward induction in MA-MDP described in Section \ref{S2} using equations (\ref{bell}). Each $u_t$ represents the value that agent $n$ is indifferent between long-term reward under actions of $a_n(t)\notin I_n(t)$ and $\bar{a}_n(t)$ at time $t$. Compared to the myopic agents' constant exploration threshold $\mu$, non-myopic agents increase their threshold to \(u_t\) for any $t<T$, targeting better options for future exploitation and larger long-term reward in (\ref{nonmyoobj}). 

The sequence $\{u_t\}_{t=1}^T$ is decreasing in $t$, as agents should explore less for shorter futures. Moreover, non-myopic agents allocate more exploration time slots in (\ref{Non_exNum}) than myopic agents' (\ref{M_exNum}). In fact, while myopic agents with effective communications tend to under-explore under the centralized communications policy, non-myopic agents selfishly choosing not to share can only rely on their own and thus over-explore at the equilibrium. Although non-myopic agents prioritize long-term rewards and intensify their exploration efforts, the overall performance of social welfare remains poor due to their practice of withholding information. Later in Section \ref{S5}, we will show that our decentralized communication mechanism obviously outperforms the centralized policy.  

\subsection{Optimal Decentralized Communication Mechanism}
Lemma \ref{sharing} still holds similarly for any decentralized communication, where agents will only truthfully share information in the last communication time slot under any $\mathcal{M}$. Then in the Pareto optimum, any $\mathcal{M}$ reduces to include only a single communication time slot $T_1$ due to non-myopic agents' information withholding. 

We denote the exploration threshold sequence to $\{u_t(T_1)\}_{t=1}^{T}$ since it depends on $T_1$ in the reduced one-time communication mechanism $\mathcal{M}_2=[T]\setminus\{T_1\}$. The centralized communications policy is a special case of $\mathcal{M}_2$ with $T_1=T-1$ for outputting the monotonic threshold sequence in (\ref{utfull}). The threshold sequence $\{u_t(T_1)\}_{t=1}^{T}$ for any $T_1$ can be solved similarly using backward induction under equation (\ref{bell}) in the MA-MDP. The key difference between \(\{u_t(T-1)\}_{t=1}^{T}\) and \(\{u_t(T_1)\}_{t=1}^{T}\) for any $T_1$ is that \(\{u_t(T_1)\}_{t=1}^{T}\) may not monotonically decrease. Intuitively, agents might halt exploration before $T_1$ to take advantage of others' option sharing at $T_1$, and resume exploration after $T_1$ if the shared information does not meet their expectations.

\begin{proposition}\label{L4}
For any one-time communication mechanism $M_2=[T]\setminus\{T_1\}$, the non-myopic agents' exploration threshold sequence $\{u_t(T_1)\}_{t=1}^{T}$ only satisfies monotonicity before and after the communication time $T_1$: $u_1(T_1)>u_2(T_1)>\ldots>u_{T_1}(T_1)$ and $u_{T_1+1}(T_1)>u_{T_1+2}(T_1)>\ldots>u_T(T_1)$, where $u_t(T_1)$ with $t>T_1$ may exceed $u_t(T_1)$ with $t\leq T_1$ (see Fig. \ref{fthres} as an illustration), which are unique solutions to a system of following equations:
\begin{align}
&u_t(T_1)-\mu= \nonumber\\
&\begin{cases}
    &(T_1-t)\left(\int^1_{u_t(T_1)} 1-F(r)dr\right)\\
&+(T-T_1)\int^1_{u_t(T_1)}G_{T_1}(r)^{N-1}(1-F(r))dr,\\
&\qquad\qquad\quad\qquad\quad\quad\quad\; \text{if 
 }t=1,\ldots,T_1,u_t(T_1)>u_{T_1+1}(T_1),\\
&(T_1-t)\left(\int^1_{u_t(T_1)} 1-F(r)dr\right)\\
&+(T-T_1)\int^{1}_{u_{T_1+1}(T_1)}G_{T_1}(r)^{N-1}(1-F(r))dr\\
    &+\sum^{T_1+k-1}_{j=T_1+1}(T-j)\int^{u_{j}(T_1)}_{u_{j+1}(T_1)}G_{T_1}(r)^{N-1} F(r)^{j-T_1}(1-F(r))dr\\
    &+(T-T_1-k)\int^{u_{T_1+k}(T_1)}_{u_t(T_1)}G_{T_1}(r)^{N-1}F(r)^{k}(1-F(r))dr,\\
&\!\!\quad\quad\quad\quad\quad\text{if 
 } t=1,\ldots,T_1,u_{T_1+k}(T_1)>u_t(T_1)>u_{T_1+k+1}(T_1),\\
 &\!\!\;\;\; \quad\quad\quad\quad\quad k=1,\ldots,T-T_1-1,\\
\end{cases}\label{Mt<T_1+1}\\
&u_t(T_1)-\mu=(T-t)\bigg(\int^1_{u_t(T_1)} 1-F(r)dr\bigg),\;\text{if }t=T_1+1,\ldots,T,\label{singlethreshold}
\end{align}
where $G_{T_1}(r)$ is given in equation (\ref{belief}) by substituting $T-1$ with $T_1$ and replacing $u_t$ with $u_t(T_1)$. The resultant social welfare over the entire time horizon is
\begin{align}
&V_{\mathcal{M}_2}=N\mu+NT_1\int^1_{u_1(T_1)}rdF(r)+N\sum^{T_1}_{t=1}F(u_{t}(T_1))^{t}\mu+\nonumber\\
    &+N\sum^{T_1-1}_{t=1}\left(F(u_t(T_1))^t\int^1_{u_t(T_1)}rdF(r)+\int^{u_t(T_1)}_{u_{t+1}(T_1)}rdF(r)^{t+1}\right)(T_1-t)\!\!\!\nonumber\\
    &+N(T-T_1)\left(1-\int^1_{u_{T_1+1}(T_1)}G_{T_1}(r)^Ndr\right)\nonumber\\
    &-N\sum^{T}_{t=T_1+2}(T-t+1)\int^{u_{t-1}(T_1)}_{u_t(T_1)}G_{T_1}(r)^NF(r)^{t-T_1-1}dr,\label{eqVN}
\end{align}
Then agents spend less expected exploration time below than (\ref{Non_exNum}):
\begin{align}
    1+\sum^{T_1}_{t=1}F(u_{t})^{t}+\sum^{T}_{t=T_1+1}G_{T_1}(u_{t})^NF(u_{t})^{t-T_1},\label{T1time}
\end{align}
where $G_{T_1}(\cdot)$ is the belief on $m_n(T_1)$ determined by $\{u_t(T_1)\}_{t=1}^{T_1}$. 
\end{proposition}

The proof of Proposition \ref{L4} is given in Appendix \ref{Appendix12}. Fig. \ref{fthres} plots $\{u_t(T_1)\}_{t=1}^T$ as a function of $T_1$ for $T$=50 and $N$=100 agents. We use Beta distribution with mean reward $\mu=0.1$ per option. After communication at $t=T_1$, agents are on their own and the MA-MDP process reduces to the single-agent MDP, where all agents follow a decreasing exploration threshold sequence $\{\bar{u}_t\}_{t=T_1+1}^T$. Before $T_1$, agents hope to hear from each other's sharing at $T_1$ and are more reluctant to explore themselves as getting closer to $T_1$. As such, their exploration thresholds $\{u_t(T_1)\}_{t=1}^{T_1}$ are less than (\ref{utfull}) before $T_1$. A well-chosen $T_1$ helps achieve a shorter exploration time (\ref{T1time}) than (\ref{Non_exNum}) to mitigate over-exploration and enhance social welfare.

\begin{figure}
    \centering
        \includegraphics[width=0.27\textwidth]{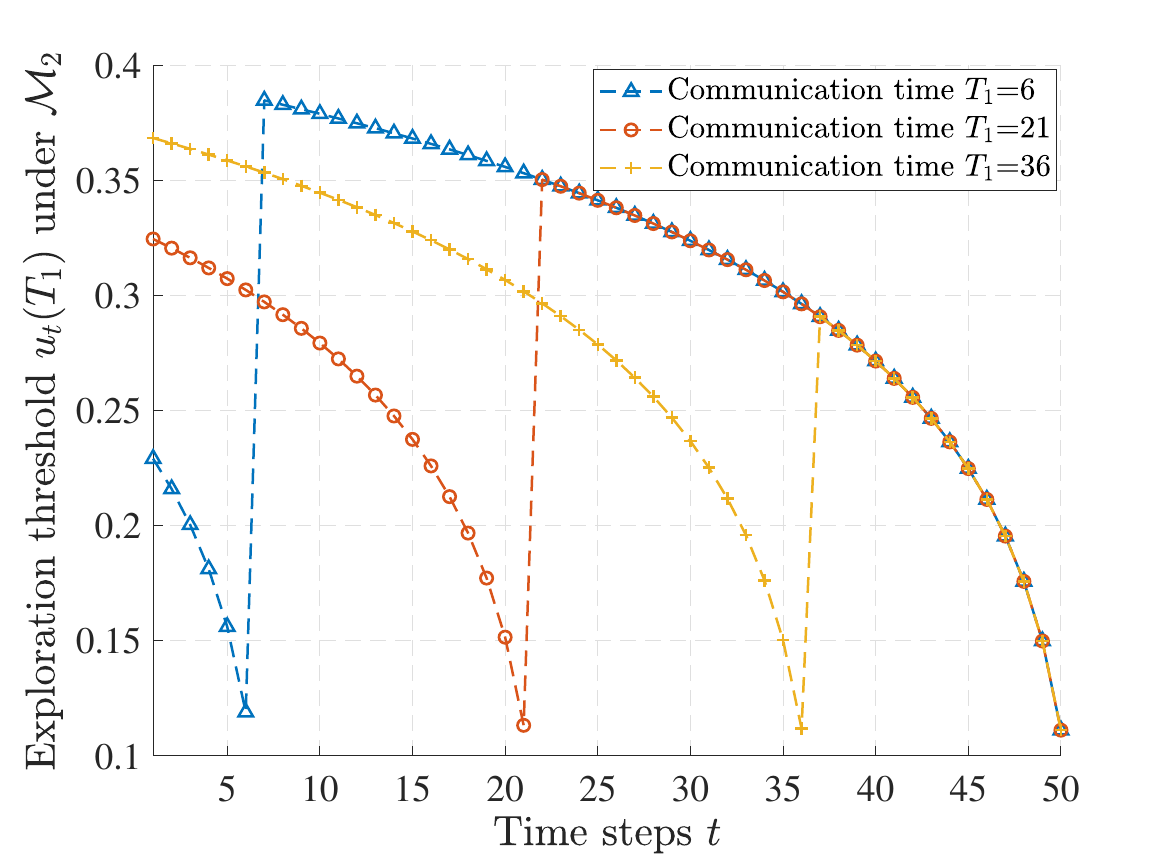}
        \caption{ Exploration threshold sequence $\{u_t(T_1)\}_{t=1}^{T}$ under different decisions of one-communication time $T_1$.}
    \label{fthres}
\end{figure}


Given $\{u_t(T_1)\}_{t=1}^{T}$, we are ready to present Algorithm \ref{A2} to maximize long-term social welfare $V_{\mathcal{M}_2}$ for obtaining the optimal mechanism $\mathcal{M}_2^*=[T]\setminus\{T_1^*\}$. To identify $T_1^*$, Algorithm \ref{A2}'s step 3-7 scan through all possible $T_1$ values to compare their corresponding social welfare. 
We still need to solve the exploration threshold sequence for any $T_1$. The equations that determine \(\{\bar{u}_t\}_{t=T_1+1}^T\) are defined by \(T - T_1\) independent equations, whereas the equations that determine \(\{\bar{u}_t(T_1)\}_{t=1}^{T_1}\) consist of \(T_1\) interdependent fixed-point equations. Hence, to initialize, in step 1 of Algorithm \ref{A2}, we first solve the degenerated single-agent MDP with $T$ equations regardless of $T_1$ to return benchmark exploration threshold sequence $\{\bar{u}_t\}_{t=1}^T$ with linear complexity $O(T)$. For each iteration of given $T_1$ from step 3 of Algorithm \ref{A2}, we solve a system of $T_1$ nonlinear equations for returning $\{u_t(T_1)\}_{t=1}^{T_1}$ before communication. By choosing $\{\bar{u}_t\}_{t=1}^{T_1}$ as the initial solution in step 4, Newton's method can efficiently and stably solve the final $\{u_t(T_1)\}_{t=1}^{T_1}$. Since the Jacobian matrix of the equations for $\{u_t(T_1)\}_{t=1}^{T_1}$ is a diagonal matrix, solving $\{u_t(T_1)\}_{t=1}^{T_1}$ only yields the complexity $O(T)$.

\begin{algorithm}
\caption{Optimal One-time Communication Mechanism $\mathcal{M}_2^*$}\label{A2}
\begin{algorithmic}[1]
\REQUIRE $T$, $N$, $F(\cdot)$
\STATE Initialize $T_1^{*}=0$, $V^* = 0$
\STATE Solve benchmark $\{\bar{u}_t\}_{t=1}^{T}$ in single-agent MDP with using Newton's method
\FOR{$T_1 = 1$ to $T$}
    \STATE Solve $\{u_t(T_1)\}_{t=1}^{T_1}$ in multi-agent MDP with initial solution $\{\bar{u}_t\}_{t=1}^{T_1}$ using Newton's method
    \STATE Calculate $V_{\mathcal{M}_2}$ in (\ref{eqVN}) under $\{u_t(T_1)\}_{t=1}^{T_1}$ and $\{u_t(T_1)\}_{t=T_1+1}^{T} = \{\bar{u}_t\}_{t=T_1+1}^{T}$
    \IF{$V_{\mathcal{M}_2} > V^*$}
        \STATE $T_1^{*} = T_1$, $\mathcal{M}^*_2 = [T] \setminus \{T_1^*\}$, $V^* = V_{\mathcal{M}^*_2}$
    \ENDIF
\ENDFOR
\RETURN $\mathcal{M}_2^* = [T] \setminus \{T_1^*\}$
\end{algorithmic}
\end{algorithm}

The performance of our algorithm is stated as follows.
\begin{theorem}\label{T2}
Algorithm \ref{A2} returns the optimal one-time communication mechanism $\mathcal{M}_2^*=[T]\setminus\{T_1^{*}\}$ in polynomial time-complexity $O(T^2)$.
\end{theorem}
The proof of Theorem \ref{T2} and the details of Algorithm \ref{A2} is given in Appendix \ref{Appendix10}.
\section{Simulation Experiments with Real Dataset for Stochastic Rewards and Heterogeneous Preference}\label{S5}
In this section, we use real hotel rating dataset from Booking.com platform to simulate our HILL system of interactive hotel customers (agents) for sharing reviews. We additionally extend to consider stochastic reward and agents’ heterogeneous preference for the same hotel. Since we have presented the experiments for myopic agents in Section \ref{S3}, and the hotel customers are found non-myopic in many instances \cite{yilmaz2022strategic}\cite{cui2018pricing}, we only conduct simulation with our Algorithm \ref{A2} for non-myopic agents in this section. We use public dataset of 825 hotels' rating statistics by customers in the popular tourism city Athens in traveling seasons. We let hotel customers (agents) post their hotel reviews (i.e., option rewards) and observe others' reviews, if the platform allows communication on that particular day.

We use the review dataset to train the prior distribution $F(\cdot)$ of option rewards (hotel ratings) in the normalized continuous range $[0,1]$. Fig.~\ref{Ffr} shows the average rating distribution of the 825 hotels with the mean $\mu$=0.49, where the red curve is the approximated PDF trained using Kernel density estimation. Then we integrate it to return CDF $F(\cdot)$ as inputs for social welfare $V_{\mathcal{M}_2}$ and Algorithm \ref{A2} for computing mechanism $\mathcal{M}_2^*=[T]\setminus \{T_1^*\}$. 

When running Algorithm \ref{A2}, Fig. \ref{f32} plots $V_{\mathcal{M}_2}/N$ per agent versus any one-communication time $T_1$ decision with the time horizon of $T=50$ and different numbers of agents $N$=20, 30, 50 per day. There, we highlight its optimum $T_1^*$ in 3 star markers returned by Algorithm \ref{A2}. Note that the centralized communications policy is a special case of our mechanism by choosing $T_1=49$ on the most right-hand side of Fig.~\ref{f32}, which leads to the lowest social welfare of value 41.8 per agent under any number of agents. Moreover, our mechanism $T_1^*$ with social welfare of above 46.7 per agent outperforms this benchmark by at least 12\% for any $N$. In Fig. \ref{f32}, optimal $T_1^*$ is delayed with the increasing number of agents $N$, as an increased $N$ facilitates the discovery of high-reward options, consequently advancing our timing of communication to share and exploit those options. Unlike centralized communications, the increase in agents boosts the social welfare per agent in our decentralized communication mechanism, which more efficiently enables early communication among non-myopic agents.

\begin{figure}
    \centering
    \begin{minipage}[b]{.23\textwidth}
        \includegraphics[width=\textwidth]{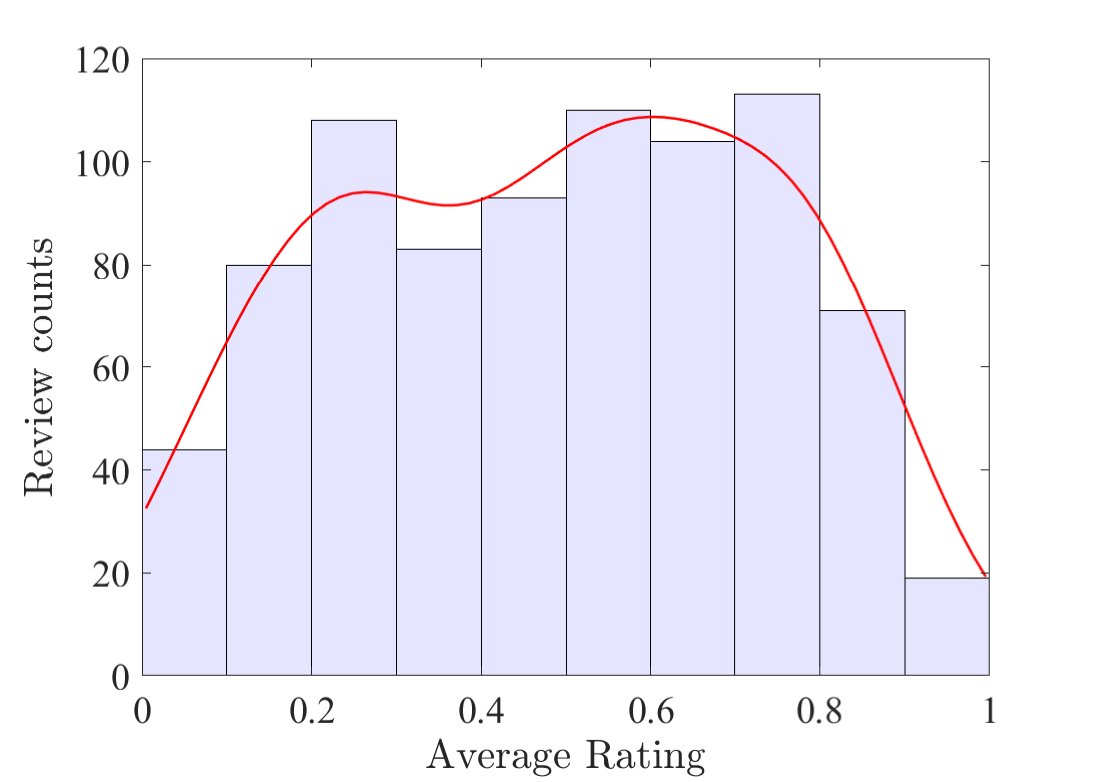}
         \caption{Distribution of hotels' average rating in a normalized range [0,1] in booking.com.}
        \label{Ffr}
    \end{minipage}\hfill
    \begin{minipage}[b]{.23\textwidth}
        \includegraphics[width=\textwidth]{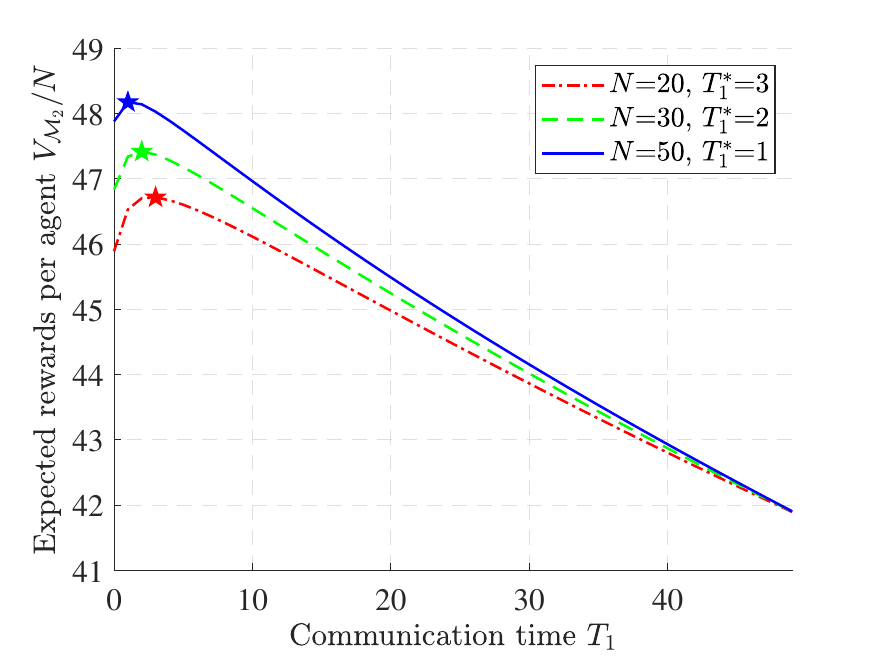}
        \caption{$V_{\mathcal{M}_2}/N$ per agent across different $T_1$ decisions with $T=50$ and agent counts of $N$=20, 30 and 50 per day.}
        \label{f32}
    \end{minipage}
    \hfill 
\end{figure}

To simulate the stochastic reward case, we add i.i.d. noise to the observed reward when agents sample each option. While non-myopic agents might use different MAB algorithms (e.g., UCB) for stochastic rewards, with relatively small noise of standard deviation 0.1 in the range of [0,1], their strategy will still converge to the optimal threshold strategy analyzed in this paper. Therefore, users are still considered to use the same strategy for deterministic rewards to approximate the non-myopic agents' strategy under stochastic rewards with small noise. To further simulate heterogeneous option reward observations among heterogeneous agents, we obtain the diverse rating data for the same hotel from different customers, and apply such data to create and add agent-heterogeneous noise as random variable to the average reward sampled from CDF $F(\cdot)$. We then employ the Monte Carlo method to simulate the HILL system with stochastic reward and heterogeneous preference for non-myopic agents, under both centralized communications policy and our decentralized communication mechanism in Alg. \ref{A2} for $T=80$ and $N=30$ per day. Conducting the experiment 500 times, we take the average of these runs as our simulation results. 

Fig. \ref{Fvr} demonstrates how the average reward per agent on the platform evolves at each time step $t\in\{0,\ldots,80\}$, under the optimal one-time communication time $T_1^*=4$ by Algorithm \ref{A2}. Before $T_1^*$, we observe that our mechanism already yields a much higher reward per time slot than centralized communications under both the stochastic reward and heterogeneous preference cases. Our mechanism just uses 4 days of no-communication for agents' self-exploration, allowing them to anticipate communication at $T_1^*$, thus avoiding over-exploration. Subsequently, agents exchange information at $T_1^*$ to exploit the very best option onwards. However, under the centralized communications policy, non-myopic agents do not share truthful information until the last time stage, converging slowly to find the high-reward option at the very last round.  

We also run the experiments for different time horizons $T$=10 to 80 days under heterogeneous preference, stochastic reward, and our original model with identical preference and fixed reward. We plot the performance gain of our mechanism over the centralized policy for the three cases in Fig. \ref{f31}. Our mechanism has obviously positive performance gain, which improves with $T$, in all three cases. Perhaps surprisingly, for $T\leq 50$, our mechanism's performance gain under the heterogeneous preference setting is higher than the original model with identical preference. As agents become more heterogeneous, they are more likely to find high-reward options and our mechanism of encouraging early-stage option sharing to exploit performs better. As $T$ increases to exceed 50, the impact of heterogeneous preference accumulates over a long time and thus our Algorithm \ref{A2} designed for the no-noise case performs slightly worse after adding the noise. The mechanism's performance gain under stochastic rewards is slightly lower than the original model's gain, but the impact is minor under small perturbations, demonstrating our mechanism's robustness to some extent of stochastic noise.



\begin{figure}
    \centering
    \begin{minipage}[b]{.23\textwidth}
            \includegraphics[width=\textwidth]{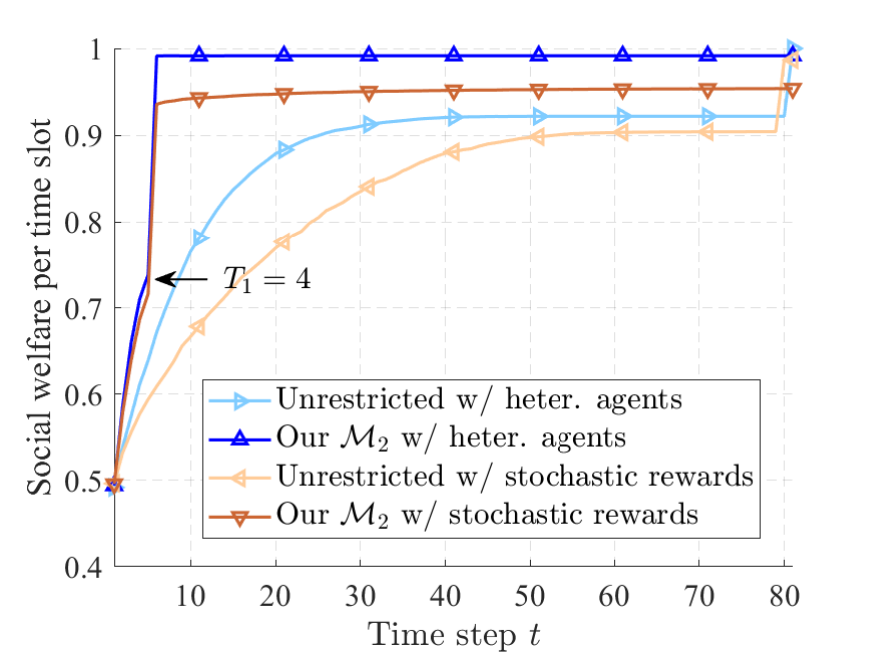}
        \caption{Average reward per agent in each time slot under our mechanism (choosing optimal $T_1^*=4$) and the centralized communications policy, where agent heterogeneity and stochastic arms are simulated.}
        \label{Fvr}

    \end{minipage}
    \hfill 
        \begin{minipage}[b]{.23\textwidth}
        \includegraphics[width=\textwidth]{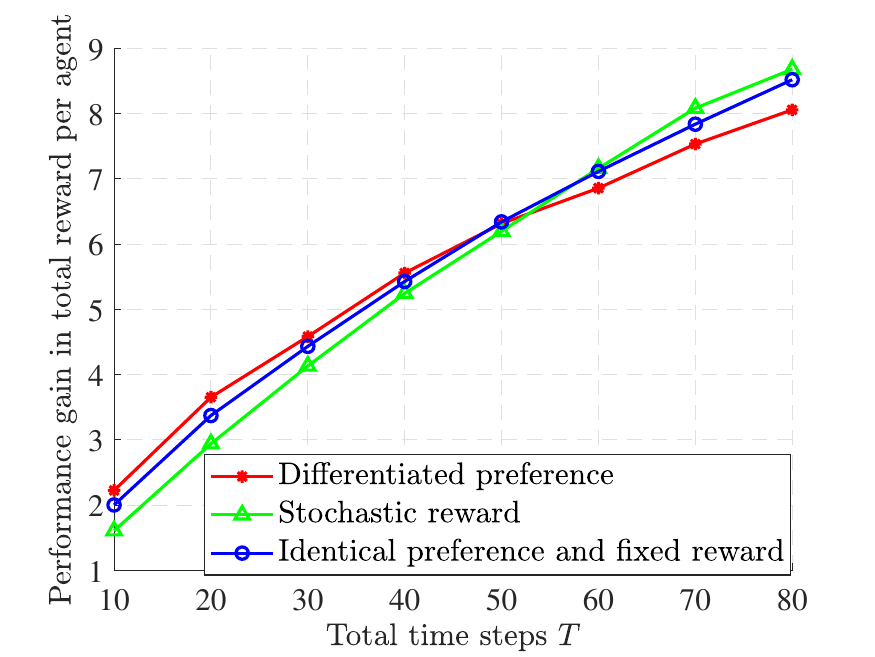}
         \caption{Performance gain per agent of our mechanism beyond the centralized communications policy under heterogeneous preference, stochastic reward, and our original identical preference and fixed reward setting, versus total time horizon $T$. Here $N=50$.}
        \label{f31}
    \end{minipage}
\end{figure}

\section{Conclusions}\label{S6}
In this paper, we propose a paradigm shift from the current centralized information sharing system to a decentralized system that only allows intermittent information sharing to regulate HILL and efficient communication among self-interested agents with deterministic rewards. By formulating the problem as a MA-MDP, we derive exploration and communication strategies for myopic and non-myopic agents. For myopic agents, we identify conditions for decentralized communication to outperform centralized communication which causes under-exploration. The NP-hard problem of maximizing long-term social welfare via decentralized communication is addressed with an asymptotically optimal algorithm of linear complexity to time communication disabling. For non-myopic agents who avoid sharing for personal gain, leading to over-exploration, we propose an optimal one-time communication mechanism with polynomial complexity. Finally, we use a real-world dataset in simulations to verify our mechanism's effectiveness, considering stochastic rewards and agents' heterogeneous option preferences.
\bibliographystyle{ACM-Reference-Format}
\bibliography{acmart}

\onecolumn


\appendix
\section{Proof of Lemma \ref{myopic_S}}\label{appendix1}

We first establish the expression for \(V_{\emptyset}\) in (\ref{fullM_socialwelfare}) and then derive the expected number of explorations in (\ref{M_exNum}). To begin, let \(E_{(t_1,t_2)}\) denote the event that an option with a reward exceeding \(\mu\) is identified during the time slots \(\{t_1,\ldots,t_2\}\), and \(\bar{E}_{t}\) denote the event that no agent has found an option with a reward greater than \(\mu\) by time \(t\). By definition, the events \(E_{(0,t)}\) and \(\bar{E}_t\) are complementary.

Define \(v_{\emptyset}(t)\) as the expected reward at time \(t\) for myopic agents under the full communication policy. Since the reward at time \(t\) depends on the exploration history up to \(t-1\), we partition the history into the cases when no option exceeding \(\mu\) has been found (event \(\bar{E}_{t-1}\)) and when at least one such option has been identified (event \(E_{(0,t-1)}\)). Thus, we can write
\begin{align}
    v_{\emptyset}(t)&=\mathbb{E}\Big[\sum^N_{n=1}r_{a_n(t)}\Big]
    =\mathbb{P}(\bar{E}_{t-1})\mathbb{E}\Big[\sum^N_{n=1}r_{a_n(t)}\Big|\bar{E}_{t-1}\Big]
    +\mathbb{P}(E_{(0,t-1)})\mathbb{E}\Big[\sum^N_{n=1}r_{a_n(t)}\Big|E_{(0,t-1)}\Big].\nonumber
\end{align}

Since all agents use an exploration threshold \(u_1=\ldots=u_t=\mu\), if no option with reward greater than \(\mu\) has been discovered (i.e., if \(\bar{E}_{t-1}\) occurs), then every agent explores at time \(t\) and earns an expected reward of \(\mu\). Hence,
\begin{align}
    \mathbb{P}(\bar{E}_{t-1})\mathbb{E}\Big[\sum^N_{n=1}r_{a_n(t)}\Big|\bar{E}_{t-1}\Big]=NF(\mu)^{Nt}\mu.\label{noone}
\end{align}

Once an option with a reward greater than \(\mu\) is found for the first time at some time \(\tau\), all agents cease further exploration. Thus, the event \(E_{(0,t-1)}\) can be decomposed as
\begin{align}
    E_{(0,t-1)}=E_{(0,0)}+\sum^{t-1}_{\tau=1}\bar{E}_{\tau-1}\cap E_{(\tau,\tau)}=\sum^{t-1}_{\tau=0}E_{(\tau,\tau)}.\label{temp}
\end{align}

Therefore, we can express the contribution from \(E_{(0,t-1)}\) as
\begin{align*}
    \mathbb{P}(E_{(0,t-1)})\mathbb{E}\Big[\sum^N_{n=1}r_{a_n(t)}\Big|E_{(0,t-1)}\Big]
    &=\mathbb{P}\Big(\sum^{t-1}_{\tau=0}E_{(\tau,\tau)}\Big)
    \mathbb{E}\Big[\sum^N_{n=1}r_{a_n(t)}\Big|\sum^{t-1}_{\tau=0}E_{(\tau,\tau)}\Big]\\[1mm]
    &=\sum^{t-1}_{\tau=0}\mathbb{P}(E_{(\tau,\tau)})\mathbb{E}\Big[\sum^N_{n=1}r_{a_n(t)}\Big| E_{(\tau,\tau)}\Big].
\end{align*}

Under the event \(E_{(\tau,\tau)}\), an option with a reward greater than \(\mu\) is identified by time \(\tau\) and is subsequently selected for exploitation at any \(t>\tau\), and its value is not affected by exploration outcomes prior to \(\tau\). Therefore, the expected social welfare at time \(t\) conditioned on \(E_{(\tau,\tau)}\) is given by
\begin{align*}
   \mathbb{E}\Big[\sum^N_{n=1}r_{a_n(t)}\Big| E_{(\tau,\tau)}\Big]
   = N\mathbb{E}\Big[\max_{n\in\mathcal{N}}\{m_n(\tau)\}\Big| E_{(\tau,\tau)}\Big]
   =N\int^1_{\mu}rd\frac{F(r)^N-F(\mu)^{N}}{1-F(\mu)^{N}}.
\end{align*}

Substituting this back into the decomposition in (\ref{temp}), we obtain
\begin{align}
    &\sum^{t-1}_{\tau=0}\mathbb{P}(E_{(\tau,\tau)})\mathbb{E}\Big[\sum^N_{n=1}r_{a_n(t)}\Big| E_{(\tau,\tau)}\Big]
    =\left(\sum^{t-1}_{\tau=0}\mathbb{P}(E_{(\tau,\tau)})\right)N\int^1_{\mu}rd\frac{F(r)^N-F(\mu)^{N}}{1-F(\mu)^{N}}\nonumber\\[1mm]
    =&(1-F(\mu)^{Nt})N\int^1_{\mu}rd\frac{F(r)^N-F(\mu)^{N}}{1-F(\mu)^{N}}.\label{eone}
\end{align}

By summing (\ref{noone}) and (\ref{eone}), we obtain the expected reward at time \(t\):
\begin{align}
    v_{\emptyset}(t)&=NF(\mu)^{Nt}\mu+(1-F(\mu)^{Nt})N\int^1_{\mu}rd\frac{F(r)^N-F(\mu)^{N}}{1-F(\mu)^{N}}\nonumber\\[1mm]
    &=N\int^1_{\mu}rd\frac{F(r)^N-F(\mu)^{N}}{1-F(\mu)^{N}}-NF(\mu)^{Nt}\left(1-\mu-\int^1_{\mu}\frac{F(r)^N-F(\mu)^{N}}{1-F(\mu)^{N}}dr\right)\nonumber\\[1mm]
    &=N\int^1_{\mu}rd\frac{F(r)^N-F(\mu)^{N}}{1-F(\mu)^{N}}-NF(\mu)^{Nt}\int^1_{\mu}\frac{1-F(r)^N}{1-F(\mu)^{N}}dr.\label{fullv}
\end{align}

Summing \(v_{\emptyset}(t)\) over the entire time horizon, the total social welfare under the full communication policy is
\begin{align*}
    V_{\emptyset}=\sum^{T}_{t=0}v_{\emptyset}(t)
    =N(T+1)\int^1_{\mu}rd\frac{F(r)^N-F(\mu)^{N}}{1-F(\mu)^{N}}-N\frac{1-F(\mu)^{N(T+1)}}{1-F(\mu)^{N}}\int^1_{\mu}\frac{1-F(r)^N}{1-F(\mu)^{N}}dr,
\end{align*}
which completes the proof of (\ref{fullM_socialwelfare}).

Let \(\tau \in \{1,\ldots,T+1\}\) denote the random variable representing the time at which myopic agents cease exploration under the full communication policy throughout the entire process. Then the expected exploration time is given by
\begin{align*}
\mathbb{E}[\tau]=\sum^{T+1}_{t=1}t\mathbb{P}(\tau=t)=\sum^{T+1}_{t=1}\mathbb{P}(\tau\geq t)=\sum^{T+1}_{t=1}F(\mu)^{N(t-1)}=\frac{1-F(\mu)^{N(T+1)}}{1-F(\mu)^N},
\end{align*}
where the second equality follows from the identity \(\mathbb{E}[\tau]=\sum_{t\ge1}\mathbb{P}(\tau\ge t)\) and the last equality results from summing the finite geometric series. This completes the proof of (\ref{M_exNum}).

\section{Proof of Proposition \ref{p1}}\label{Appendix11}
We have shown that equation (\ref{eq_lemmq3.2}) is sufficient for the non-optimality of a full communication policy in the presence of myopic agents, as discussed in Section \ref{S3}. It remains to demonstrate that (\ref{eq_lemmq3.2}) is also a necessary condition. To this end, consider equation (\ref{VM}) for \(V_{\mathcal{M}}\) under any communication restriction mechanism. If \(V_{\mathcal{M}} > V_{\emptyset}\), then at least one of the terms
\begin{align*}
    (T-T_m-\Delta_m)y_{\Delta_m+1}-\sum_{i=1}^{\Delta_m}x_i
\end{align*}
for some \(m\in\{1,\ldots,M\}\) must be positive. This implies that
\begin{align*}
    (T-\Delta_m)y_{\Delta_m+1}-\sum_{i=1}^{\Delta_m}x_i>0,
\end{align*}
which is exactly the condition given in (\ref{eq_lemmq3.2}).

\section{Proof for Lemma \ref{L333}}\label{Appendix16}
First, we prove the expressions for \(V_{\bar{\mathcal{M}}}\) in (\ref{specialM}) for the special communication restriction mechanism \(\bar{\mathcal{M}}=\{T_1,\ldots,T_1+\Delta_1-1\}\), which features a single no-communication window. Next, we extend the result to the general case as stated in Lemma \ref{L333}. In particular, we first prove the expression for \(V_{\mathcal{M}}\) in (\ref{VM}), and then we establish the formula for the expected number of exploration time slots in (\ref{M_exNum_R}).

 Given $E_{(t_1,t_2)}$ and $\bar{E}_t$ as defined in Appendix \ref{appendix1}, we introduce $E^n_{(t_1,t_2)}$ and $\bar{E}^n_t$, customizing these events for a specific agent $n$ as required for the analysis in this section. We introduce the following two lemmas:
\begin{lemma}\label{event1}
    Under the event \(E_{(T_m,T_m+\Delta_m)}\) for \(m=1,\ldots,M\), in which an option with a reward higher than \(\mu\) is identified during the no-communication interval \(t\in\{T_m,\ldots,T_m+\Delta_m\}\) for the long-term exploitation of all myopic agents, the expected maximum reward discovered after this window is\footnote{Note that communication always occurs after selecting an action at each time slot. Although the no-communication time window is defined as \(\{T_m,\ldots,T_m+\Delta-1\}\), the action chosen at time \(T_m+\Delta_m\) is also under no-communication, resulting in the same impact of communication restriction as for \(\{T_m,\ldots,T_m+\Delta-1\}\).}
    \begin{align}
        &\mathbb{P}(E_{(T_m,T_m+\Delta_m)})\mathbb{E}[\max_{n\in\mathcal{N}}m_{n}(T_m+\Delta_m)|E_{(T_m,T_m+\Delta_m)}]\nonumber\\
    =&F(\mu)^{NT_m}\left(1-\mu F(\mu)^{N(\Delta_m+1)}-\int^1_{\mu}\left(\frac{F(r)-F(\mu)}{1-F(\mu)}+F(\mu)^{\Delta_m+1}\frac{1-F(r)}{1-F(\mu)}\right)^Ndr\right).\label{ti}
    \end{align}
\end{lemma}
\begin{proof}
    Note that the event \(E_{(T_m,T_m+\Delta_m)}\) implies that no agent discovers an option with a reward higher than \(\mu\) during the time interval from 0 to \(T_m-1\). Let \(b_k\) denote the event that exactly \(k\) agents find an option with a reward exceeding \(\mu\) within the window \(\{T_m, \ldots, T_m+\Delta_m\}\). Then, the event \(E_{(T_m,T_m+\Delta_m)}\) can be further decomposed as follows: \begin{align*}
    E_{(T_m,T_m+\Delta_m)}=\sum^N_{k=1} b_k.
\end{align*}
Therefore,
\begin{align}
    \mathbb{P}(E_{(T_m,T_m+\Delta_m)})\mathbb{E}[\max_{n\in\mathcal{N}}m_{n}(T_m+\Delta_m)|E_{(T_m,T_m+\Delta_m)}]
    =&\sum^N_{k=1}\mathbb{P}(b_k)\mathbb{E}[\max_{n\in\mathcal{N}}m_{n}(T_m+\Delta_m)|b_k]\nonumber\\
    =&\sum_{k=1}^NF(\mu)^{NT_m}(1-F(\mu)^{\Delta_m+1})^kF(\mu)^{(\Delta_m+1)(N-k)}\begin{pmatrix}
        N\\k
    \end{pmatrix}\int^1_{\mu}rd\left(\frac{F(r)-F(\mu)}{1-F(\mu)}\right)^k\nonumber\\
    =&F(\mu)^{NT_m}\int^1_{\mu}rd\left(\left(\frac{F(r)-F(\mu)}{1-F(\mu)}(1-F(\mu)^{\Delta_m+1})\right)+F(\mu)^{\Delta_m+1}\right)^N-F(\mu)^{N(\Delta_m+1)}\nonumber\\
    =&F(\mu)^{NT_m}\left(1-\mu F(\mu)^{N(\Delta_m+1)}-\int^1_{\mu}\left(\frac{F(r)-F(\mu)}{1-F(\mu)}+F(\mu)^{\Delta_m+1}\frac{1-F(r)}{1-F(\mu)}\right)^Ndr\right).\label{et2}
\end{align}
\end{proof}
\begin{lemma}\label{event2}
Under event \( E_{(T_m+\Delta_m+1,T_{m+1}-1)} \) for \( m = 1,\dots,M \), if \( T_{m+1} > T_m+\Delta_m+1 \), an option yielding a reward greater than \(\mu\) is identified and exploited by all myopic agents during the interval \( \{T_m+\Delta_m+1,\dots,T_{m+1}-1\} \). If \( T_{m+1} = T_m+\Delta_m+1 \), this event has probability zero and cannot occur. The expected maximum reward discovered after this interval is therefore given by:
\begin{align}
    \mathbb{P}(E_{(T_m+\Delta_m+1,T_{m+1}-1)})\mathbb{E}[\max_{n\in\mathcal{N}}m_{n}(T_{m+1}-1)|E_{(T_m+\Delta_m+1,T_{m+1}-1)}] = &(F(\mu)^{N(T_m+\Delta_m+1)}-F(\mu)^{NT_{m+1}}) \int_{\mu}^{1} r d\frac{F(r)^N-F(\mu)^N}{1-F(\mu)^N},\label{tii}
\end{align}
which equals 0 when \( T_{m+1} = T_m + \Delta_m + 1 \). Therefore, this expression applies for any \( T_m \), \( \Delta_m \), and \( T_{m+1} \) defined by the mechanism.
\end{lemma}

\begin{proof}
In communication time window $\in\{T_m+\Delta_m+1,,\ldots,T_{m+1}-1\}$, all agents will immediately stop exploration at time $\tau$ under event $E_{(T_m+\Delta_m+1,T_{m+1}-1)}$. Then
    \begin{align*}
        \mathbb{P}(E_{(T_m+\Delta_m+1,T_{m+1}-1)})\mathbb{E}[\max_{n\in\mathcal{N}}m_{n}(T_{m+1}-1)|E_{(T_m+\Delta_m+1,T_{m+1}-1)}]
        =&\sum^{T_{m+1}-1}_{\tau=T_m+\Delta_m+1}\mathbb{P}(E_{(\tau,\tau)})\mathbb{E}[\max_{n\in\mathcal{N}}m_{n}(\tau)|E_{(\tau,\tau)}]\\
        =&\sum^{T_{m+1}-1}_{\tau=T_m+\Delta_m+1}\mathbb{P}(E_{(\tau,\tau)})\int^1_{\mu}dr\frac{F(r)^N-F(\mu)^N}{1-F(\mu)^N}\\
        =&(F(\mu)^{N(T_m+\Delta_m+1)}-F(\mu)^{NT_{m+1}})\int^1_{\mu}rd\frac{F(r)^N-F(\mu)^N}{1-F(\mu)N}.
    \end{align*}
\end{proof}

Under the communication restriction mechanism $\bar{\mathcal{M}}=\{T_1,\ldots,T_1+\Delta_1-1\}$, the total time slots can be segmented into three time windows as follows:
\begin{align}
    \{0,\ldots,T_1\}, \{T_1+1,\ldots,T_1+\Delta_1\}, \{T_1+\Delta_1+1,\ldots,T\}.\label{windows}
\end{align}

The expected social welfare \( v_{\bar{\mathcal{M}}}(t) \) under mechanism \( \bar{\mathcal{M}} \) at time \( t \) depends on which time window in (\ref{windows}) the time \( t \) belongs to. We analyze \( v_{\bar{\mathcal{M}}}(t) \) separately for each window. For \( t \in \{0,\dots,T_1-1\} \), the expected social welfare \( v_{\bar{\mathcal{M}}}(t) \) matches the welfare under full communication, \( v_{\emptyset}(t) \) in (\ref{fullv}), because both mechanisms share the same historical event space during this interval.

For $t\in\{T_1+1,\ldots,T_1+\Delta_1\}$, the historical events preceding $t\in\{T_1+1,\ldots,T_1+\Delta_1\}$ for agent $n$ are categorized as follows:
\begin{align}
    E_{(0,T_1-1)},E^n_{(T_1,t-1)},\bar{E}^n_{t-1},\label{partition1}
\end{align}
under which $v_{\bar{\mathcal{M}}}(t)$ for $t\in\{T_1+1,\ldots,T_1+\Delta_1\}$ is given by:
\begin{align}
    v_{\bar{\mathcal{M}}}(t)=\mathbb{P}(E_{(0,T_1-1)})\mathbb{E}[\sum_{n=1}^Nr_{a_n(t)}|E_{(0,T_1-1)}]+\sum_{n=1}^N\mathbb{P}(E^n_{(T_1,t-1)})\mathbb{E}[r_{a_n(t)}|E^n_{(T_1,t-1)}]+\sum_{n=1}^N\mathbb{P}(\bar{E}^n_{t-1})\mathbb{E}[r_{a_n(t)}|\bar{E}^n_{t-1}],\label{vbarmzong}
\end{align}
where the component $\mathbb{P}(E_{(0,T_1-1)})\mathbb{E}[\sum_{n=1}^Nr_{a_n(t)}|E_{(0,T_1-1)}]$, signifying the expected reward for long-term exploitation under event $E_{(0,T_1-1)}$, can be directly generalized from (\ref{eone}) in Appendix \ref{appendix1} as:
\begin{align}
    \mathbb{P}(E_{(0,T_1-1)})\mathbb{E}[\sum_{n=1}^Nr_{a_n(t)}|E_{(0,T_1-1)}]=N(1-F(\mu)^{NT_1})\int^1_{\mu}rd\frac{F(r)^N-F(\mu)^{N}}{1-F(\mu)^{N}}.\label{genefull}
\end{align}
Regarding the component $\mathbb{P}(E^n_{(T_1,t-1)})\mathbb{E}[r_{a_n(t)}|E^n_{(T_1,t-1)}]$ in (\ref{vbarmzong}), the event $E^n_{(T_1,t-1)}$ indicates that agent $n$ independently discovers an option with a reward exceeding $\mu$, leading to the following expected reward for exploitation:
\begin{align}
    \mathbb{P}(E^n_{(T_1,t-1)})\mathbb{E}[r_{a_n(t)}|E^n_{(T_1,t-1)}]=\mathbb{P}(E^n_{(T_1,t-1)})\int^1_{\mu}rd\frac{F(r)-F(\mu)}{1-F(\mu)}=F(\mu)^{NT_1}(1-F(\mu)^{t-T_1})\int^1_{\mu}rd\frac{F(r)-F(\mu)}{1-F(\mu)}.\label{T_mt}
\end{align}
Conversely, under the event \(\bar{E}_{t-1}^n\), agent \(n\) must choose exploration at time \(t\), yielding an expected reward of \(F(\mu)^{N t}\mu\). Combining these three components, we compute \( v_{\bar{\mathcal{M}}}(t) \) for \( t \in \{T_1+1,\dots,T_1+\Delta_1\} \) as follows:
\begin{align*}
    v_{\bar{\mathcal{M}}}(t)=&N(1-F(\mu)^{NT_1})\int^1_{\mu}rd\frac{F(r)^N-F(\mu)^{N}}{1-F(\mu)^{N}}+NF(\mu)^{NT_1}(1-F(\mu)^{t-T_1})\int^1_{\mu}rd\frac{F(r)-F(\mu)}{1-F(\mu)}+NF(\mu)^{NT_1+t-T_1}\mu\\
    =&v_{\emptyset}(t)+NF(\mu)^{NT_1}\int^1_{\mu}\left(\frac{F(r)^N-F(\mu)^N}{1-F(\mu)^N}+F(\mu)^{N(t-T_1)}\frac{1-F(r)^N}{1-F(\mu)^N}\right)-\left(\frac{F(r)-F(\mu)}{1-F(\mu)}+F(\mu)^{t-T_1}\frac{1-F(r)}{1-F(\mu)}\right)dr\\
    =&v_{\emptyset}(t)-NF(\mu)^{NT_1}x_{t-T_1},
\end{align*}
where $x_{t-T_1}$ is given in (\ref{xi}).

Next, we examine \( v_{\bar{\mathcal{M}}}(t) \) for 
\[
t \in \{T_1+\Delta_1+1,\ldots,T\},
\]
where the historical events preceding \( t \) are categorized as follows:
\[
\begin{aligned}
    &E_{(0,T_1-1)},\; E_{(T_1,T_1+\Delta_1)},\; E_{(T_1+\Delta_1+1,t-1)},\; \bar{E}_{t-1},\quad \text{if }t>T_1+\Delta_1+1,\\[1em]
    &E_{(0,T_1-1)},\; E_{(T_1,T_1+\Delta_1)},\; \bar{E}_{t-1},\quad \text{if }t=T_1+\Delta_1+1,
\end{aligned}
\]
where the expected reward for exploitation associated with \( E_{(0,T_1-1)} \),
\(
\mathbb{P}(E_{(0,T_1-1)})\mathbb{E}[r_{a_n(t)} \mid E_{(0,T_1-1)}],
\)
is given in (\ref{genefull}). Similarly, the expected reward for exploitation related to \( E_{(T_1,T_1+\Delta_1)} \),
\(
\mathbb{P}(E_{(T_1,T_1+\Delta_1)})\mathbb{E}[r_{a_n(t)} \mid E_{(T_1,T_1+\Delta_1)}],
\)
is provided in (\ref{ti}) of Lemma \ref{event1}. Additionally, conditioned on the event \( \bar{E}_{t-1} \), the expected reward at time \( t \) is the exploration reward \( \mu \), which gives
\[
\mathbb{P}(\bar{E}_{t-1})\mathbb{E}[r_{a_n(t)} \mid \bar{E}_{t-1}] = F(\mu)^{Nt}\mu.
\]
It remains to assess the expected reward for exploitation under the event \( E_{(T_1+\Delta_1+1,t-1)} \). This can be directly generalized from (\ref{tii}) in Lemma \ref{event2} (since \( t>T_1+\Delta_1+1 \)) as follows:
\begin{align}
    \mathbb{P}(E_{(T_1+\Delta_1+1,t-1)})\mathbb{E}[r_{a_n(t)} \mid E_{(T_1+\Delta_1+1,t-1)}] 
    = \bigl(F(\mu)^{N(T_1+\Delta_1+1)}-F(\mu)^{Nt}\bigr)
    \int^1_{\mu} r\,d\frac{F(r)^N-F(\mu)^N}{1-F(\mu)^N},\label{full222}
\end{align}
By combining (\ref{genefull}), (\ref{ti}) with \( T_m=T_1 \) and \( \Delta_m=\Delta_1 \), \( F(\mu)^{Nt}\mu \), and (\ref{full222}), we derive the final form of \( v_{\bar{\mathcal{M}}}(t) \) for 
\(
t \in \{T_1+\Delta_1+1,\ldots,T\}.
\)
\begin{align*}
    v_{\bar{\mathcal{M}}}(t)
    =&N(1-F(\mu)^{NT_1})\int^1_{\mu}rd\frac{F(r)^N-F(\mu)^{N}}{1-F(\mu)^{N}}+NF(\mu)^{NT_1}\left(1-\mu F(\mu)^{N(\Delta_1+1)}-\int^1_{\mu}\left(\frac{F(r)-F(\mu)}{1-F(\mu)}+F(\mu)^{\Delta_1+1}\frac{1-F(r)}{1-F(\mu)}\right)^Ndr\right)\\
    &+N(F(\mu)^{N(T_1+\Delta_1+1)}-F(\mu)^{Nt})\int^1_{\mu}rd\frac{F(r)^N-F(\mu)^N}{1-F(\mu)^N}+NF(\mu)^{Nt}\\
    =&v_{\emptyset}(t)+NF(\mu)^{NT_1}\int^1_{\mu}\left(\frac{F(r)^N-F(\mu)^N}{1-F(\mu)^N}+F(\mu)^{N(\Delta_1+1)}\frac{1-F(r)^N}{1-F(\mu)^N}\right)-\left(\frac{F(r)-F(\mu)}{1-F(\mu)}+F(\mu)^{\Delta_1+1}\frac{1-F(r)}{1-F(\mu)}\right)^Ndr\\
    =&v_{\emptyset}(t)+NF(\mu)^{NT_1}y_{\Delta_1+1}.
\end{align*}
where $y_{{\Delta_1}+1}$ is given in (\ref{yi}). Integrate $v_{\bar{\mathcal{M}}}(t)$ for $t$ in different time window, we have
\begin{align*}
    v_{\bar{\mathcal{M}}}(t)=\begin{cases}
        v_{\emptyset}(t),&\text{if } t=1,\ldots,T_1,\\
        v_{\emptyset}(t)-NF(\mu)^{NT_1}x_{t-T_1},&\text{if } t=T_1+1,\ldots,T_1+\Delta_1,\\
        v_{\emptyset}(t)+NF(\mu)^{NT_1}y_{\Delta_1+1},&\text{if } t=T_1+\Delta_1+1,\ldots,T.
    \end{cases}
\end{align*}
Then, the social welfare for total time horizon $V_{\bar{\mathcal{M}}}$ is given by:
\begin{align*}
    V_{\bar{\mathcal{M}}}&=\sum^{T_1}_{t=0}v_{\bar{\mathcal{M}}}(t)+\sum^{T_1+\Delta_1}_{t=T_1+1}v_{\bar{\mathcal{M}}}(t)+\sum^{T}_{t=T_{1}+\Delta_1+1}v_{\bar{\mathcal{M}}}(t)=V_{\emptyset}+N(T-T_1-\Delta_1)F(\mu)^{NT_1}y_{\Delta_1+1}-NF(\mu)^{NT_1}\sum^{\Delta_1}_{t=1}x_{t},
\end{align*}
which completes the proof of (\ref{specialM}). 

Next, we consider the expected reward $v_{{\mathcal{M}}}(t)$ at time $t$ for general communication restriction mechanism $\mathcal{M}=\cup_{m=1}^M\{T_m,\ldots,T_m+\Delta_m-1\}$. Generalizing from (\ref{partition1}), when \(t\) falls within \(\{T_m,\ldots,T_m+\Delta_m\}\), the historical events can be partitioned as follows:
\begin{align*}
    E_{(0,T_m-1)}, E_{(T_m,T_m+\Delta_m)},E_{(T_m+\Delta_m+1,T_{m+1}-1)},\ldots, E^n_{(T_m,t-1)},\bar{E}^n_{t-1},
\end{align*}
under which $v_{{\mathcal{M}}}(t)$ for $t\in\{T_m,\ldots,T_m+\Delta_m\}$ is
\begin{align}
    v_{\mathcal{M}}(t)=&N\mathbb{P}(E_{(0,T_1-1)})\mathbb{E}[r_{a_n(t)}|E_{(0,T_1-1)}] + \sum^{m-1}_{i=1}N\mathbb{P}(E_{(T_i,T_i+\Delta_1)})\mathbb{E}[r_{a_n(t)}|E_{(T_i,T_i+\Delta_1)}] \nonumber\\
    &+ \sum^{m-1}_{i=1}N\mathbb{P}(E_{(T_i+\Delta_i+1,T_{i+1}-1)})\mathbb{E}[r_{a_n(t)}|E_{(T_i+\Delta_i+1,T_{i+1}-1)})] + \sum^N_{n=1}(\mathbb{P}(E^n_{(T_m,t-1)})\mathbb{E}[r_{a_n(t)}|E^n_{(T_m,t-1)}] + \mathbb{P}(\bar{E}^n_{t-1})\mathbb{E}[r_{a_n(t)}|\bar{E}^n_{t-1}]), \label{c}
\end{align}
where each part can be directly generalized from the derivation of $v_{\bar{\mathcal{M}}}(t)$. Then
\begin{align}
    v_{\mathcal{M}}(t)=&N(1-F(\mu)^{NT_1})\int^1_{\mu}rd\frac{F(r)^N-F(\mu)^N}{1-F(\mu)^N}+NF(\mu)^{NT_m}(1-F(\mu)^{t-T_m})\int^1_{\mu}rd\frac{F(r)-F(\mu)}{1-F(\mu)}+F(\mu)^{NT_m+t-T_m}\mu\nonumber\\
    &+N\sum^{m-1}_{i=1}F(\mu)^{NT_i}\left(1-\mu F(\mu)^{N(\Delta_i+1)}-\int^1_{\mu}\left(\frac{F(r)-F(\mu)}{1-F(\mu)}+F(\mu)^{\Delta_i+1}\frac{1-F(r)}{1-F(\mu)}\right)^Ndr\right)\nonumber\\
    &+N\sum^{m-1}_{i=1}(F(\mu)^{N(T_i+\Delta_i+1)}-F(\mu)^{NT_{i+1}})\int^1_{\mu}rd\frac{(F(r)^N-F(\mu)^N)}{1-F(\mu)N}\nonumber\\
    =&N(1-F(\mu)^{NT_m})\int^1_{\mu}rd\frac{F(r)^N-F(\mu)^N}{1-F(\mu)^N}+NF(\mu)^{NT_m}(1-F(\mu)^{t-T_m})\int^1_{\mu}rd\frac{F(r)-F(\mu)}{1-F(\mu)}+F(\mu)^{NT_m+t-T_m}\mu\nonumber\\
    &+N\sum^{m-1}_{i=1}F(\mu)^{NT_i}\left(1-\mu F(\mu)^{N(\Delta_i+1)}-\int^1_{\mu}\left(\frac{F(r)-F(\mu)}{1-F(\mu)}+F(\mu)^{\Delta_i+1}\frac{1-F(r)}{1-F(\mu)}\right)^Ndr\right)\nonumber\\
    &+N\sum^{m-2}_{i=1}(F(\mu)^{N(T_i+\Delta_i+1)}-F(\mu)^{NT_{i+1}})\left(1-\int^1_{\mu}\frac{(F(r)^N-F(\mu)^N)}{1-F(\mu)N}\right)dr\nonumber\\
    &+N(F(\mu)^{N(T_{m-1}+\Delta_{m-1}+1)}-F(\mu)^{NT_1})\left(1-\int^1_{\mu}\frac{(F(r)^N-F(\mu)^N)}{1-F(\mu)N}\right)dr.\label{pp}
\end{align}
The first line in (\ref{pp}) is: 
\begin{align}
    &N(1-F(\mu)^{NT_m})\int^1_{\mu}rd\frac{F(r)^N-F(\mu)^N}{1-F(\mu)^N}+NF(\mu)^{NT_m}(1-F(\mu)^{t-T_m})\int^1_{\mu}rd\frac{F(r)-F(\mu)}{1-F(\mu)}+NF(\mu)^{NT_m+t-T_m}\mu\nonumber\\
    =&v_{\emptyset}(t)-N\int^1_{\mu}rd\frac{F(r)^N-F(\mu)^N}{1-F(\mu)^N}+NF(\mu)^{Nt}\int^1_{\mu}\frac{1-F(r)^N}{1-F(\mu)^N}dr+\nonumber\\
    &N(1-F(\mu)^{NT_m})\int^1_{\mu}rd\frac{F(r)^N-F(\mu)^N}{1-F(\mu)^N}+NF(\mu)^{NT_m}(1-F(\mu)^{t-T_m})\int^1_{\mu}rd\frac{F(r)-F(\mu)}{1-F(\mu)}+F(\mu)^{NT_m+t-T_m}\mu\nonumber\\
    =&v_{\emptyset}(t)+NF(\mu)^{NT_m}\left(-\int^1_{\mu}\frac{F(r)-F(\mu)}{1-F(\mu)}dr+\int^1_{\mu}\frac{F(r)^N-F(\mu)^N}{1-F(\mu)^N}dr\right)\nonumber\\
    &+NF(\mu)^{NT_m}\left(F(\mu)^{N(t-T_m)}\int^1_{\mu}\frac{1-F(r)^N}{1-F(\mu)^N}dr-F(\mu)^{t-T_m}\int^1_{\mu}\frac{1-F(r)}{1-F(\mu)}dr\right)\nonumber\\
    =&v_{\emptyset}(t)-NF(\mu)^{NT_m}x_{t-T_m}.\label{agg1}
\end{align}
Then the remaining parts of (\ref{pp}) are
\begin{align}
&N\sum^{m-1}_{i=1}F(\mu)^{NT_i}\left(1-\mu F(\mu)^{N(\Delta_i+1)}-\int^1_{\mu}\left(\frac{F(r)-F(\mu)}{1-F(\mu)}+F(\mu)^{\Delta_i+1}\frac{1-F(r)}{1-F(\mu)}\right)^Ndr\right)\nonumber\\
    &+N\sum^{m-2}_{i=1}(F(\mu)^{N(T_i+\Delta_i+1)}-F(\mu)^{NT_{i+1}})\left(1-\int^1_{\mu}\frac{(F(r)^N-F(\mu)^N)}{1-F(\mu)^N}dr\right)+N(F(\mu)^{N(T_{m-1}+\Delta_{m-1}+1)}-F(\mu)^{NT_1})\left(1-\int^1_{\mu}\frac{(F(r)^N-F(\mu)^N)}{1-F(\mu)^N}dr\right)\nonumber\\
    =&N\sum^{m-1}_{i=1}F(\mu)^{NT_{i}}\left(\int^1_{\mu}\frac{F(r)^N-F(\mu)^N}{1-F(\mu)^N}+F(\mu)^{N(\Delta_i+1)}\frac{1-F(r)^N}{1-F(\mu)^N}-\left(\frac{F(r)-F(\mu)}{1-F(\mu)}+F(\mu)^{\Delta_i+1}\frac{1-F(r)}{1-F(\mu)}\right)^Ndr\right)\nonumber\\
    =&N\sum^{m-1}_{i=1}F(\mu)^{NT_{i}}y_{\Delta_i+1}.\label{agg2}
\end{align}
Finally, summing up (\ref{agg1}) and (\ref{agg2}), $v_{\mathcal{M}}(t)$ for $t\in\{T_m+1,\ldots,T_m+\Delta_m\}$ is
\begin{align}
    v_{\mathcal{M}}(t)=v_{\emptyset}(t)+N\sum^{m-1}_{i=1}F(\mu)^{NT_{i}}y_{\Delta_i+1}-NF(\mu)^{NT_m}x_{t-T_m}, t\in\{T_m+1,\ldots,T_m+\Delta_m\}, m=1,\ldots,M.
\end{align}

For $t\in\{T_m+\Delta_m+1,\ldots,T_{m+1}\}$, the historical events can be partitioned as follows:
\begin{align*}
    E_{(0,T_m-1)}, E_{(T_m,T_m+\Delta_m)},E_{(T_m+\Delta_m+1,T_{m+1}-1)},\ldots, E_{(T_m+\Delta_m+1,t-1)},\bar{E}_{t-1},
\end{align*}
under which $v_{{\mathcal{M}}}(t)$ for $t\in\{T_m+\Delta_m+1,\ldots,T_{m+1}\}$ is
\begin{align}
    &v_{\mathcal{M}}(t)\\
    =&N\mathbb{P}(E_{(0,T_1-1)})\mathbb{E}[r_{a_n(t)}|E_{(0,T_1-1)}] + \sum^{m}_{i=1}N\mathbb{P}(E_{(T_i,T_i+\Delta_1)})\mathbb{E}[r_{a_n(t)}|E_{(T_i,T_i+\Delta_1)}] + N\sum^{m-1}_{i=1}\mathbb{P}(E_{(T_i+\Delta_i+1,T_{i+1}-1)})\mathbb{E}[r_{a_n(t)}|E_{(T_i+\Delta_i+1,T_{i+1}-1)})] \nonumber\\
    &+ N\mathbb{P}(E_{(T_m+\Delta_m+1,t-1)})\mathbb{E}[r_{a_n(t)}|E_{(T_m+\Delta_m+1,t-1)}] + \mathbb{P}(\bar{E}_{t-1})\mathbb{E}[r_{a_n(t)}|\bar{E}_{t-1}]\nonumber\\
    =&N(1-F(\mu)^{NT_1})\int^1_{\mu}rd\frac{F(r)^N-F(\mu)^N}{1-F(\mu)^N}+N(F(\mu)^{N(T_m+\Delta_m+1)}-F(\mu)^{Nt})\int^1_{\mu}rd\frac{F(r)^N-F(\mu)^N}{1-F(\mu)N}+F(\mu)^{Nt}\mu\nonumber\\
    &+N\sum^{m}_{i=1}F(\mu)^{NT_i}\left(1-\mu F(\mu)^{N(\Delta_i+1)}-\int^1_{\mu}\left(\frac{F(r)-F(\mu)}{1-F(\mu)}+F(\mu)^{\Delta_i+1}\frac{1-F(r)}{1-F(\mu)}\right)^Ndr\right)\nonumber\\
    &+N\sum^{m-1}_{i=1}(F(\mu)^{N(T_i+\Delta_i+1)}-F(\mu)^{NT_{i+1}})\int^1_{\mu}rd\frac{(F(r)^N-F(\mu)^N)}{1-F(\mu)N}\nonumber\\
    =&v_{\emptyset}(t)+N\sum^{m}_{i=1}F(\mu)^{N(T_i+\Delta_i+1)}\int^1_{\mu}\frac{1-F(r)^N}{1-F(\mu)^N}dr+N\sum^{m}_{i=1}F(\mu)^{NT_i}\left(\int^1_{\mu}\frac{F(r)^N-F(\mu)^N}{1-F(\mu)^N}dr-\int^1_{\mu}\left(\frac{F(r)-F(\mu)}{1-F(\mu)}+F(\mu)^{\Delta_i+1}\frac{1-F(r)}{1-F(\mu)}\right)^Ndr\right)\nonumber\\
    =&v_{\emptyset}(t)+N\sum^{m}_{i=1}F(\mu)^{NT_i}y_{\Delta_i+1},
\end{align}
For social welfare of total time horizon $V_{\mathcal{M}}$, we calculate it by summing up $v_{\mathcal{M}}(t)$ above for all time slots:
\begin{align}
    V_{\mathcal{M}}=&N\sum^{T_1}_{t=0}v_{\mathcal{M}}(t)+\sum_{m=1}^M\left(\sum^{T_m+\Delta_m}_{t=T_m+1}v_{\mathcal{M}}(t)+N\sum^{T_{m+1}}_{t=T_m+\Delta_m+1}v_{\mathcal{M}}(t)\right)\nonumber\\
    =&V_{\emptyset}+N\sum_{m=1}^M\left(\sum^{T_m+\Delta_m}_{t=T_m+1}\left(\sum^{m-1}_{i=1}F(\mu)^{NT_{i}}y_{\Delta_i+1}-F(\mu)^{NT_m}x_{t-T_m}\right)+\sum^{T_{m+1}}_{t=T_m+\Delta_m+1}\sum^{m}_{i=1}F(\mu)^{NT_i}y_{\Delta_i+1}\right)\nonumber\\
    =&V_{\emptyset}+N\sum_{m=1}^M(T-T_m-\Delta_m)F(\mu)^{NT_m}y_{\Delta_m+1}-N\sum_{m=1}^M\sum^{T_m+\Delta_m}_{t=T_m+1}F(\mu)^{NT_m}x_{t-T_m}\nonumber\\
    =&V_{\emptyset}+N\sum_{m=1}^MF(\mu)^{NT_m}\left(T-T_m-\Delta_m)y_{\Delta_m+1}-\sum^{\Delta_m}_{i=1}x_{i}\right),
\end{align}
which completes the proof of (\ref{VM}). 

Let $\tau$, in $\{1,\ldots,T+1\}$, denote the random variable representing exploration time by the myopic agents under communication restriction mechanism $\mathcal{M}=\cup_{m=1}^M\{T_m,\ldots,T_m+\Delta_m-1\}$ in the total process. Then the expectation of $\tau$ is
\begin{align*}
\mathbb{E}[\tau]=&\sum^{T}_{t=1}t\mathbb{P}(\tau=t)=\sum^{T}_{t=1}\mathbb{P}(\tau\geq t)=\sum^{T_1+1}_{t=1}\mathbb{P}(\tau\geq t)+\sum^M_{m=1}\sum^{T_m+\Delta_m+1}_{t=T_m+2}\mathbb{P}(\tau\geq t)+\sum^M_{m=1}\sum_{t=T_m+\Delta_m+2}^{T_{m+1}+1}\mathbb{P}(\tau\geq t)\\
=&\sum^{T_1+1}_{t=1}F(\mu)^{N(t-1)}+\sum^M_{m=1}\sum^{T_m+\Delta_m+1}_{t=T_m+2}F(\mu)^{NT_m+t-T_m-1}+\sum^M_{m=1}\sum_{t=T_m+\Delta_m+2}^{T_{m+1}+1}F(\mu)^{N(t-1)}\\
=&\sum^{T_1}_{t=0}F(\mu)^{Nt}+\sum^M_{m=1}\left(\sum^{\Delta_m}_{t=1}F(\mu)^{NT_m+t}+\sum_{t=T_m+\Delta_m+1}^{T_{m+1}}F(\mu)^{Nt}\right)
\end{align*}
which completes the proof of (\ref{M_exNum_R}).

\section{Proof of Theorem \ref{T3.6}}\label{Appendix15}
    Denote $\tilde{\mathcal{M}}^*_1=\{0,\ldots,\tilde{\Delta}^*_1-1\}$ as the solution returned by Algorithm \ref{A1}, and $\mathcal{M}^*=\{0,\ldots,\Delta_1^*-1\}\cup_{m=1}^{M-1}\{\sum^m_{i=1}(\Delta_i^*+1),\ldots,\sum^m_{i=1}(\Delta_i^*+1)+\Delta^*_{m+1}-1\}$ as the optimal solution of problem (\ref{OPT}). Recall that $\tilde{\mathcal{M}}^*_1$ is the best single communication restriction window mechanism, which provide us the following inequality
\begin{align}
    &\left((T-\sum_{i=1}^m\Delta_i^*-m+1)y_{\Delta_m^*+1}-\sum_{i=1}^{\Delta_m^*}x_i\right)\leq (T-\tilde{\Delta}^*_1)y_{\tilde{\Delta}^*_1+1}-\sum_{i=1}^{\tilde{\Delta}^*_1}x_i, \forall m\in\{1,\ldots,M\},\label{Inthe}
\end{align}
which is derived from
    \begin{align*}
        & (T-\sum_{i=1}^m\Delta_i^*-m+1)y_{\Delta_m^*+1}-\sum_{i=1}^{\Delta_m^*}x_i<(T-\Delta_m^*)y_{\Delta_m^*+1}-\sum_{i=1}^{\Delta_m^*}x_i<\left((T-\tilde{\Delta}^*_1)y_{\tilde{\Delta}^*_1+1}-\sum_{i=1}^{\tilde{\Delta}^*_1}x_i\right).
    \end{align*}
Then we have
\begin{align*}
    \frac{V_{\mathcal{M}^*}}{V_{\tilde{\mathcal{M}}^*_1}}&=\frac{\left((T-\Delta_1^*)y_{\Delta_1^*+1}-\sum_{i=1}^{\Delta_1^*}x_i\right)+\sum^M_{m=2}F(\mu)^{N\sum_{i=1}^{m-1}(\Delta_i^*+1)}\left((T-\sum_{i=1}^m\Delta_i^*-m+1)y_{\Delta_m^*+1}-\sum_{i=1}^{\Delta_m^*}x_i\right)}{\left((T-\tilde{\Delta}^*_1)y_{\tilde{\Delta}^*_1+1}-\sum_{i=1}^{\tilde{\Delta}^*_1}x_i\right)}\nonumber\\
    &<\frac{(1+\sum^M_{m=2}F(\mu)^{N\sum_{i=1}^{k-1}(\Delta_i^*+1)})\left((T-\tilde{\Delta}^*_1)y_{\tilde{\Delta}^*_1+1}-\sum_{i=1}^{\tilde{\Delta}^*_1}x_i\right)}{\left((T-\tilde{\Delta}^*_1)y_{\tilde{\Delta}_1+1}-\sum_{i=1}^{\tilde{\Delta}^*_1}x_i\right)}\nonumber\\
    &=1+\sum^M_{m=2}F(\mu)^{N\sum_{i=1}^{m-1}(\Delta_i^*+1)}\nonumber\\
    &\leq 1+\sum^{ T/2}_{m=2}F(\mu)^{2N(m-1)}\\
    &=\frac{1-F(\mu)^{TN}}{1-F(\mu)^{2N}},
\end{align*}
where the first inequality comes from (\ref{Inthe}), the second inequality comes from $\Delta^*_i\geq 1$ and the upper bound of $M$ in problem (\ref{pro2}). Then
\begin{align*}
    V_{\mathcal{M}^*}&<\frac{1-F(\mu)^{TN}}{1-F(\mu)^{2N}}V_{\tilde{\mathcal{M}}^*_1}\\
    V_{\mathcal{M}^*}\left(1-\frac{F(\mu)^{2N}-F(\mu)^{TN}}{1-F(\mu)^{TN}}\right)&<V_{\tilde{\mathcal{M}}^*_1},
\end{align*}
where $\left(1-\frac{F(\mu)^{2N}-F(\mu)^{TN}}{1-F(\mu)^{TN}}\right)$ is the approximate ratio. As $N\rightarrow\infty$, if $T$ is also large enough to satisfy condition (\ref{eq_lemmq3.2}), the approximate ratio will tend to 1.

\section{Proof for Lemma \ref{NonMyFull}}\label{appendix4}
We begin by establishing the monotonicity of \(\{u_t\}_{t=1}^T\) as stated in (\ref{utfull}). With this monotonicity in hand, we then verify the equations in (\ref{utfull}). Next, we derive the expression for \(P(r)\) in (\ref{belief}), and finally, we establish the expected number of time slots for exploration as given in (\ref{Non_exNum}).

First, note that \(u_T = \mu\) and \(u_t \geq \mu\) for every \(t \in [T]\). We then show that any threshold sequence \(\{u_t\}_{t=1}^{T-1}\) containing two elements \(u_i\) and \(u_j\) with \(i < j\) and \(u_i < u_j\) can be improved upon by an alternative strategy that yields a strictly higher expected reward.

The existence of \(u_i, u_j \in \{u_t\}_{t=1}^{T-1}\) with \(i < j\) and \(u_i < u_j\) implies that there is at least one time instance \(t \in \{1, \ldots, T-1\}\) where \(u_t < u_{t+1}\). Suppose that at time \(t\), the agent \(n\)'s maximum reward is \(m_n(t)\) satisfying \(u_t < m_n(t) < u_{t+1}\). According to the threshold strategy, the agent exploits at time \(t\) and explores at time \(t+1\), yielding a combined expected reward of \(m_n(t) + \mu\) over these two time slots. Under these actions, the expected maximum reward \(m_n(t+2)\) after time \(t+1\) is given by:
\begin{align}
    m_n(t+2)=\int^1_{m_n(t)} r\,dF(r) + F(m_n(t))\,m_n(t).\label{emm}
\end{align}
Alternatively, if the agent explores at time \(t\) and exploits at time \(t+1\), the expected reward from \(t\) to \(t+1\) becomes:
\[
\mu + \int^1_{m_n(t)} r\,dF(r) + F(m_n(t))\,m_n(t),
\]
which is strictly greater than \(m_n(t) + \mu\) under the original threshold strategy, because it leverages the possibility of obtaining a higher reward at \(t\). Moreover, the resulting \(m_n(t+2)\) remains as in (\ref{emm}), meaning that the expected reward after time \(t+1\) is unaffected by the switch. Therefore, a threshold strategy that includes \(u_i, u_j \in \{u_t\}_{t=1}^{T-1}\) with \(i < j\) and \(u_i < u_j\) cannot be optimal. Since \(\{u_t\}_{t=1}^T\) is the optimal solution for problem (\ref{nonmyoobj}) under the full communication policy, this completes the proof of the monotonicity of \(\{u_t\}_{t=1}^T\).

Next, we prove the equation of \(\{u_t\}_{t=1}^T\) in (\ref{utfull}). At each time \(t\), \(u_t\) represents the state \(m_n(t)\) at which agent \(n\) is indifferent between exploring and exploiting. Note that \(u_t > u_{t+1}\) implies that if \(m_n(t) = u_t > u_{t+1}\), then agent \(n\) will cease to explore from time \(t+1\) onward. Let \(R_t^n\) denote the reward for agent \(n\), which is a random variable with distribution \(F(\cdot)\), if agent \(n\) chooses to explore at time \(t\). Under this scenario, the long-term reward for agent \(n\) starting from time \(t\) when opting to explore is given by:
\begin{align*}
    \mathbb{E}\Biggl[\sum^T_{\tau=t}r_{a_n(\tau)} \,\Big|\, a_n(t)\in \mathcal{K}\setminus I_n(t),\, m_n(t)=u_t\Biggr]
    = \mu+\max\{R^n_t,m_n(t)\}(T-t)+\max\{R^n_t,m_n(t),\max_{i\neq n} m_i(T-1)\}.
\end{align*}
Recalling that \(P(\cdot)\) denotes the distribution of \(m_i(T-1)\), we can express this expectation as:
\begin{align}
    \mathbb{E}\Biggl[\sum^T_{\tau=t}r_{a_n(\tau)} \,\Big|\, a_n(t)\in \mathcal{K}\setminus I_n(t),\, m_n(t)=u_t\Biggr]
    =&\mu+\left(\int^1_{u_t}r\,dF(r)+u_tF(u_t)\right)(T-t-1)\nonumber\\
    &\quad+\int^1_{u_t}r\,dF(r)P(r)^{N-1}+u_tF(u_t)P(u_t)^{N-1}\nonumber\\
    =&\mu+\left(1-\int^1_{u_t}F(r)\,dr\right)(T-t-1)+1-\int^1_{u_t}F(r)P(r)^{N-1}\,dr,\label{ifexplore}
\end{align}
Similarly, if agent \(n\) opts for exploitation at time \(t\), her long-term reward is:
\begin{align}
    \mathbb{E}\Biggl[\sum^T_{\tau=t}r_{a_n(\tau)} \,\Big|\, a_n(t)=\bar{a}_n(t),\, m_n(t)=u_t\Biggr]
    =&\, m_n(t)(T-t)+\max\{m_n(t),\max_{i\neq n} m_i(T-1)\}\nonumber\\
    =&\, u_t(T-t)+\int^1_{u_t}r\,dP(r)^{N-1}+u_tP(u_t)^{N-1}\nonumber\\
    =&\, u_t(T-t)+1-\int^1_{u_t}P(r)^{N-1}\,dr.\label{ifexploit}
\end{align}
The value \(u_t\) for \(t=1,\ldots,T-1\) is then determined by equating the expressions in (\ref{ifexplore}) and (\ref{ifexploit}):
\begin{align}
    \mu+\left(1-\int^1_{u_t}F(r)\,dr\right)(T-t-1)+1-\int^1_{u_t}F(r)P(r)^{N-1}\,dr
    &= u_t(T-t)+1-\int^1_{u_t}P(r)^{N-1}\,dr\nonumber\\
    \int^1_{u_t}P(r)^{N-1}\bigl(1-F(r)\bigr)dr+(T-t-1)\int^1_{u_t}\bigl(1-F(r)\bigr)dr
    &= u_t-\mu,\quad t=1,\ldots,T-1.\label{utttt}
\end{align}
This completes the proof for the equation of \(\{u_t\}_{t=1}^T\) in (\ref{utfull}).

Next, we will prove the distribution $P(\cdot)$ of the maximum reward $m_n(t)$ for agent $n$ at time $t$, as outlined in (\ref{belief}), when agent $n$ adheres to the threshold strategy $\{u_t\}_{t=1}^T$ described in (\ref{utfull}). We modify the definition of $E_{(t_1,t_2)}$ in Appendix \ref{appendix1} here for subsequent analysis. Let $E_t$ denote the event where agent $n$'s maximum reward $m_n(t)$ exceeds the exploration threshold, $m_n(t) > u_t$, for the first time at time $t$. Conversely, $\bar{E}_t$ indicates that the agent does not meet the exploration cessation condition at any time $\tau = 1,\ldots,T_1$. Then, $P(\cdot)$ with $u_1\leq r<1$ is:
\begin{align*}
    &P(r|u_1\leq r<1)=\mathbb{P}(m_n(T-1)<r|u_1\leq r<1)=1-\mathbb{P}(m_n(T-1)\geq r|u_1\leq r<1).
\end{align*}   
We know that $E_1,\ldots,E_{T-1}$ and $\bar{E}_{T-1}$ form the complete event space for the historical events before $T-1$. Note that under the event $E_t$ and for $r > u_1$, the event $m_n(T-1) > r$ can only occur if $R^n_{t-1} > r$, and in fact, $m_n(T-1) = R^n_{t-1}$. Then
\begin{align*}
    &1-\mathbb{P}(m_n(T-1)\geq r|u_1\leq r<1)\\
    =&1-\sum^{T-1}_{t=1}\mathbb{P}(E_t)\mathbb{P}(m(T-1)\geq r|E_t,u_1\leq r<1)-\mathbb{P}(\bar{E}_{T-1})\mathbb{P}(m(T-1)\geq r|\bar{E}_{T-1},u_1\leq r<1)\\
    =&1-(\mathbb{P}(R^n_0>r)+\mathbb{P}(R^n_0<u_1,R^n_1>r)+\mathbb{P}(R^n_0,R^n_1<u_2,R^n_1>r)+\ldots+\mathbb{P}(R^n_0,\ldots,R^n_{T-2}<u_{T-1},R^n_{T-1}>r))=\\
    =&1-(1-F(r)+F(u_1)(1-F(r))+F(u_2)^2(1-F(r))+\ldots+F(u_{T-1})^{T-1}(1-F(r)))=\\
    =&F(r)-(1-F(r))\sum^{T-1}_{t=1}F(u_t)^{t}.
\end{align*}
Similarly, $P(r)$ for $u_t\leq r<u_{t-1}$ is:
\begin{align*}
    &P(r|u_t\leq r<u_{t-1})\\
    =&\mathbb{P}(m(T-1)<r|u_t\leq r<u_{t-1})=1-\mathbb{P}(m(T-1)\geq r|u_t\leq r<u_{t-1})\\
    =&1-\sum^{T-1}_{\tau=1}\mathbb{P}(E_\tau)\mathbb{P}(m(T-1)\geq r|E_\tau,u_t\leq r<u_{t-1})-\mathbb{P}(\bar{E}_{T-1})\mathbb{P}(m(T-1)\geq r|\bar{E}_T,u_t\leq r<u_{t-1}).\\
\end{align*}  
Note that, under the event $E_{\tau}$ for $\tau=1,\ldots,t-1$ with $u_t\leq r<u_{t-1}$, the event $m(T-1)\geq r$ happens with probability 1, since agent $n$ must stop with a reward greater than $u_{\tau}>r$. Then
\begin{align*}
=&1-\sum^{t-1}_{\tau=1}\mathbb{P}(E_\tau)-\sum_{\tau=t}^{T-1}\mathbb{P}(m(T-1)\geq r|E_t,u_t\leq r<u_{t-1})-\mathbb{P}(\bar{E}_{T-1})\mathbb{P}(m(T-1)\geq r|\bar{E}_T,u_t\leq r<u_{t-1})\\
    =&1-(1-F(u_1)+F(u_1)-F(u_2)^2+\ldots+F(u_{t-2})^{t-2}-F(u_{t-1})^{t-1})\\
    &-(F(u_{t-1})^{t-1}-F(r)^{t}+F(u_{t})^t(1-F(r))+F(u_{t+1})^{t+1}(1-F(r))+\ldots+F(u_{T-1})^{T-1}(1-F(r)))\\
    =&F(r)^{t}-(1-F(r))\sum^{T-1}_{i=t}F(u_i)^i.
\end{align*}
Finally, for $P(r) = \mathbb{P}(m(T-1) <r)$ when $r < u_{T-1}$, the event $m(T-1) < r$ occurs only if all explored rewards from time 0 to $T-1$ are less than $r$. Therefore, $P(r) = F(r)^T$ for $r < u_{T-1}$. Then we complete the proof for (\ref{belief}).

Let $\tau$, in $\{1,\ldots,T+1\}$, denote the random variable representing exploration time by non-myopic agents under full communication policy in the total process. Then expectation of $\tau$ is
\begin{align*}
\mathbb{E}[\tau]=\sum^{T+1}_{t=1}t\mathbb{P}(\tau=t)=\sum^{T+1}_{t=1}\mathbb{P}(\tau\geq t)=1+\sum^{T-1}_{t=1}F(u_t)^{t}+P(u_T)^N=1+\sum^{T-1}_{t=1}F(u_t)^{t}+F(u_T)^{NT}
\end{align*}
which completes the proof of (\ref{Non_exNum}).

\section{Proof of Proposition \ref{L4}}\label{Appendix12}
First, monotonicity within the sequences $\{u_t(T_1)\}_{t=1}^{T_1}$ and $\{u_t(T_1)\}_{T_1+1}^{T}$ can be established analogously by the procedure detailed in Appendix \ref{appendix4}, since agents do not communicate prior to or after time \(T_1\). However, the relative ordering between $u_t(T_1)$ for \(t \leq T_1\) and for \(t > T_1\) cannot be demonstrated similarly, owing to the information exchange occurring precisely at time \(T_1\).

Next, we derive the equations for the threshold sequence $\{u_t(T_1)\}_{t=1}^{T_1}$ in (\ref{Mt<T_1+1}) and (\ref{singlethreshold}), and for the total social welfare $V_{\mathcal{M}_2}$ in (\ref{eqVN}), under the communication restriction mechanism $\mathcal{M}_2=[T]\setminus\{T_1\}$. The distribution $P(\cdot)$ for the maximum reward $m_n(T_1)$ at communication time $T_1$ follows directly from the analogous distribution proof for $m_n(T-1)$ in Appendix \ref{appendix4} by substituting $T_1$ for $T-1$. Thus, we omit it here.

Considering the threshold sequence $u_t(T_1)$ for $t > T_1$ as in equation (\ref{singlethreshold}), we observe it reduces to the single-agent exploration threshold, independent of $T_1$, denoted as $\{\bar{u}_t\}_{t=1}^T$. Specifically, we have $\{u_t(T_1)\}_{t=T_1+1}^T = \{\bar{u}_t\}_{t=T_1+1}^T$. Recall from Appendix \ref{appendix4} that $R_t^n$ represents the exploration reward agent $n$ receives at time $t$, drawn from distribution $F(\cdot)$ when choosing exploration. When the agent's maximum reward $m_n(t)$ equals the threshold $\bar{u}_t$, she is indifferent between exploration and exploitation. If she chooses exploration, her expected long-term reward (\ref{nonmyoobj}) becomes
\begin{align}
    \mu + \mathbb{E}[\max\{R^n_t,\bar{u}_t\}] (T - t) = \mu + \left(1 - \int^1_{\bar{u}_t} F(r) dr\right)(T - t).\label{if1}
\end{align}
If instead she selects exploitation, her expected long-term reward (\ref{nonmyoobj}) equals $\bar{u}_t(T - t + 1)$. Equating these two alternatives yields the threshold $\bar{u}_t$ equation for all $t\in\{1,\ldots,T\}$:
\begin{align*}
    \bar{u}_t - \mu &= (T - t)\left(\int^1_{\bar{u}_t} 1 - F(r)\,dr\right).
\end{align*}

Now, we derive $u_t(T_1)$ for $t\in\{1,\ldots,T_1\}$, corresponding to equations (\ref{Mt<T_1+1}). When $u_t(T_1) > u_{T_1+1}(T_1)$, fulfilling the exploration-stopping condition at time $t$ also implies fulfillment in the subsequent period $\{T_1,\ldots,T\}$. Hence, in this scenario, the derivation of $u_t(T_1)$ parallels the threshold derivation (\ref{utfull}) under full communication policy outlined in Appendix \ref{appendix4}.

Conversely, if \( u_t(T_1) < u_{T_1+1}(T_1) \), it's possible that \( u_{T_1+1}(T_1) > m_n(t) > u_t(T_1) \), indicating the agent may stop exploration at time \( t \) but resume after time \( T_1+1 \). Suppose threshold \( u_t(T_1) \) lies within the interval \([\bar{u}_{T_1+k+1}, \bar{u}_{T_1+k}]\). In this case, event \( E^n_{T_1+\tau} \), for some \(1 \leq \tau \leq k+1\), occurs with positive probability. Given this setting, if the agent’s maximum reward at time $t$ equals $u_t(T_1)$ and she chooses to explore, her expected long-term reward (\ref{nonmyoobj}) can be expressed as:
\begin{align}
    &\sum^T_{\tau=t}\mathbb{E}[r_{a_n(\tau)}|a_n(t)\in \mathcal{K}\setminus I_n(t), m_n(t)=u_t(T_1)]\nonumber\\
    =&\sum^{T_1}_{\tau=t}\mathbb{E}[r_{a_n(\tau)}|a_n(t)\in \mathcal{K}\setminus I_n(t),m_n(t)=u_t(T_1)]+\sum^{T}_{\tau=T_1+1}\mathbb{E}[r_{a_n(\tau)}|a_n(t)\in \mathcal{K}\setminus I_n(t),m_n(t)=u_t(T_1)]\nonumber\\
    =&\mu+\mathbb{E}[\max\{R^n_t,u_t(T_1)\}](T_1-t)+\mathbb{P}(E_{T_1+1})\mathbb{E}[\max\{R^n_t,\max_{i\neq n} m_i(T_1)\}|E_{T_1+1}](T-T_1)\nonumber\\
    &+\sum^{k-1}_{j=1}\mathbb{P}(E_{T_1+j+1})\bigg(\mu j+\mathbb{E}[\max\{R^n_t, u_t(T_1), \max_{i \neq n} m_i(T_1),\max_{\tau=1,\ldots,j}R^n_{T_1+\tau}\}|E_{T_1+j+1}](T-T_1-j)\bigg)\nonumber\\
    &+\mathbb{P}(E_{T_1+k+1})\bigg(\mu k+\mathbb{E}[\max\{R^n_t, u_t(T_1), \max_{i \neq n} m_i(T_1),\max_{\tau=1,\ldots,k}R^n_{T_1+\tau}\}|E_{T_1+k+1}](T-T_1-k)\bigg).\label{part4}
\end{align}
Since the expressions differ significantly by line, we proceed by evaluating each term separately. The first line can be immediately obtained from equation (\ref{if1}), as follows:
\begin{align}
   \mu+\left(1-\int^1_{u_t(T_1)}F(r)dr\right)(T_1-t).\label{lll1}
\end{align}
Then the second line is
\begin{align}
    \int^1_{\bar{u}_{T_1+1}}rdP(r)^{N-1}F(r)(T-T_1)=\bigg(1-\bar{u}_{T_1+1}P(\bar{u}_{T_1+1})^{N-1}F(\bar{u}_{T_1+1})-\int^1_{\bar{u}_{T_1+1}}P(r)^{N-1}F(r)dr\bigg)(T-T_1).\label{lll2}
\end{align}
The third line is
\begin{align}
&\sum^{k-1}_{j=1}\left[\mu j\bigg(F(\bar{u}_{T_1+j})\cdot P(\bar{u}_{T_1+j})^{N-1}F(\bar{u}_{T_1+j})^{j-1}-F(\bar{u}_{T_1+j+1})\cdot P(\bar{u}_{T_1+j+1})^{N-1}F(\bar{u}_{T_1+j+1})^{j}\bigg)\right.\nonumber\\
    &\left.+\bigg(\int^{\bar{u}_{T_1+j}}_{\bar{u}_{T_1+j+1}}rdP(r)^{N-1}F(r)\cdot F(r)^{j}+F(\bar{u}_{T_1+j})P(\bar{u}_{T_1+j})^{N-1}\cdot F(\bar{u}_{T_1+j})^{j-1}\int^1_{\bar{u}_{T_1+j}}rdF(r)\bigg)(T-T_1-j)\right]\nonumber\\
    =&\sum^{k-1}_{j=1}\left[\mu j\bigg(F(\bar{u}_{T_1+j})^jP(\bar{u}_{T_1+j})^{N-1}-F(\bar{u}_{T_1+j+1})^{j+1} P(\bar{u}_{T_1+j+1})^{N-1}\bigg)\right.\nonumber\\
    &+F(\bar{u}_{T_1+j})^jP(\bar{u}_{T_1+j})^{N-1}\bigg(1-\int^1_{\bar{u}_{T_1+j}}F(r)dr\bigg)(T-T_1-j)\nonumber\\
    &\left.-\bigg(\bar{u}_{T_1+j+1}P(\bar{u}_{T_1+j+1})^{N-1} F(\bar{u}_{T_1+j})^{j+1}+\int^{\bar{u}_{T_1+j}}_{\bar{u}_{T_1+j+1}}P(r)^{N-1}F(r)^{j+1}dr\bigg)(T-T_1-j)\right].\nonumber\\
\end{align}
Replace the bracket $\bigg(1-\int^1_{\bar{u}_{T_1+j}}F(r)dr\bigg)(T-T_1-j)$ in this equation by $(\bar{u}_{T_1+j}(T-T_1-j+1)-\mu)$ derived from (\ref{singlethreshold}), it follows that
\begin{align}
    =&\sum^{k-1}_{j=1}\left[\mu j\bigg(F(\bar{u}_{T_1+j})^jP(\bar{u}_{T_1+j})^{N-1}-F(\bar{u}_{T_1+j+1})^{j+1} P(\bar{u}_{T_1+j+1})^{N-1}\bigg)\right.+F(\bar{u}_{T_1+j})^jP(\bar{u}_{T_1+j})^{N-1}\bigg(\bar{u}_{T_1+j}(T-T_1-j+1)-\mu\bigg)\nonumber\\
    &\left.-\bigg(\bar{u}_{T_1+j+1}P(\bar{u}_{T_1+j+1})^{N-1} F(\bar{u}_{T_1+j})^{j}+\int^{\bar{u}_{T_1+j}}_{\bar{u}_{T_1+j+1}}P(r)^{N-1}F(r)^{j+1}dr\bigg)(T-T_1-j)\right]\nonumber\\
=&-\mu F(\bar{u}_{T_1+k})^{k} P(\bar{u}_{T_1+k})^{N-1}(k-1)+\bar{u}_{T_1+1}F(\bar{u}_{T_1+1})^jP(\bar{u}_{T_1+1})^{N-1}(T-T_1)\nonumber\\
&-\bar{u}_{T_1+k}F(\bar{u}_{T_1+k})^{k}P(\bar{u}_{T_1+k})^{N-1}(T-T_1-k+1)-\sum^{k-1}_{j=1}\int^{\bar{u}_{T_1+j}}_{\bar{u}_{T_1+j+1}}P(r)^{N-1}F(r)^{j+1}dr(T-T_1-j).
\label{lll3}
\end{align}
Finally, the last line is
\begin{align}
    &\mu F(\bar{u}_{T_1+k})^{k}P(\bar{u}_{T_1+k})^{N-1}k+\left(\int^{\bar{u}_{T_1+k}}_{u_t(T_1)}rdP(r)^{N-1}F(r)^{k+1}+P(\bar{u}_{T_1+k})^{N-1} F(\bar{u}_{T_1+k})^{k}\int^1_{\bar{u}_{T_1+k}}rdF(r)\right.\nonumber\\
    &+u_t(T_1)P(u_t(T_1))^{N-1}F(u_t(T_1))^{k+1}\bigg)(T-T_1-k)\nonumber\\
    =&\bar{u}_{T_1+k}F(\bar{u}_{T_1+k})^kP(\bar{u}_{T_1+k})^{N-1}(T-T_1-k+1)+\mu F(\bar{u}_{T_1+k})^kP(\bar{u}_{T_1+k})^{N-1}(k-1)-\int^{\bar{u}_{T_1+k}}_{u_t(T_1)}P(r)^{N-1}F(r)^{k+1}dr(T-T_1-k).\label{lll4}
\end{align}
By summing up these equation of each line of (\ref{part4}), we derive the final expression of the expected long-term reward for agent $n$ if she chooses to explore at time $t$:
\begin{align}
    &\mu+\left(1-\int^1_{u_t(T_1)}F(r)dr\right)(T_1-t)+\bigg(1-\int^1_{\bar{u}_{T_1+1}}P(r)^{N-1}F(r)dr\bigg)(T-T_1)\nonumber\\
    &-\sum^{k-1}_{j=1}\int^{\bar{u}_{T_1+j}}_{\bar{u}_{T_1+j+1}}P(r)^{N-1}F(r)^{j+1}dr(T-T_1-j)-\int^{\bar{u}_{T_1+k}}_{u_t(T_1)}P(r)^{N-1}F(r)^{k+1}dr(T-T_1-k).\label{Final1}
\end{align}

Next, we will determine agent $n$'s expected long-term reward in (\ref{nonmyoobj}) if she opts for exploitation at time $t$, given that $m_n(t) = u_t(T_1)$ and $u_t(T_1)\in[\bar{u}_{T_1+k+1}, \bar{u}_{T_1+k}]$. Similar to (\ref{part4}), we have
\begin{align}
    &\mathbb{E}[\sum^T_{\tau=t}r_{a_n(\tau)}|a_n(t)=\bar{a}_n(t),m_n(t)=u_t(T_1)]\nonumber\\
=&\sum^{T_1}_{\tau=t}\mathbb{E}[r_{a_n(\tau)}|a_n(t)=\bar{a}_n(t),m_n(t)=u_t(T_1)]+\sum^{T}_{\tau=T_1+1}\mathbb{E}[r_{a_n(\tau)}|a_n(t)=\bar{a}_n(t),m_n(t)=u_t(T_1)]\nonumber\\
    =&u_t(T_1)(T_1-t+1)+\mathbb{P}(E_{T_1+1})\mathbb{E}[\max\{u_t(T_1), \max_{i \neq n} m_i(T_1),\max_{\tau=1,\ldots,j}R^n_{T_1+\tau}\}|E_{T_1+1}](T-T_1)\nonumber\\
    &+\sum^{k-1}_{j=1}\mathbb{P}(E_{T_1+j+1})\bigg(\mu j+\mathbb{E}[\max\{u_t(T_1), \max_{i \neq n} m_i(T_1),\max_{\tau=1,\ldots,j}R^n_{T_1+\tau}\}|E_{T_1+j+1}](T-T_1-j)\bigg)\nonumber\\
    &+\mathbb{P}(E_{T_1+k+1})\bigg(\mu k+\mathbb{E}[\max\{u_t(T_1), \max_{i \neq n} m_i(T_1),\max_{\tau=1,\ldots,k}R^n_{T_1+\tau}\}|E_{T_1+k+1}](T-T_1-k)\bigg).\label{zong2}
\end{align}
Note that, unlike the expression in (\ref{part4}) for the long-term reward of exploration, the expected reward after \( T_1 \) lacks the \( R^n_t \) component because there is no exploration at time \( t \). Then, by following a derivation process similar to that for (\ref{part4}), we obtain equation (\ref{zong2}) as follows:
\begin{align}
    &u_t(T_1)(T_1-t+1)+\bigg(1-\int^1_{\bar{u}_{T_1+1}}P(r)^{N-1}dr\bigg)(T-T_1)-\sum^{k-1}_{j=1}\int^{\bar{u}_{T_1+j}}_{\bar{u}_{T_1+j+1}}P(r)^{N-1}F(r)^{j}dr(T-T_1-j)\nonumber\\
    &-\int^{\bar{u}_{T_1+k}}_{u_t(T_1)}P(r)^{N-1}F(r)^{k}dr(T-T_1-k).\label{Final2}
\end{align}
By equating (\ref{Final1}) with (\ref{Final2}), we derive equation (\ref{Mt<T_1+1}) when $u_t(T_1)<\bar{u}_{T_1+1}$, which complete the prove for (\ref{Mt<T_1+1}) and (\ref{singlethreshold}).

We next prove \(V_{\mathcal{M}_2}\) in (\ref{eqVN}). In addition to the event \(E_t\) defined in Appendix \ref{appendix1}, we further define \(E_{T_1+t}\) for $t=1,\ldots,T-T_1$ as the event where agent \(n\)'s maximum reward \(m_n(T_1+t)\) exceeds the exploration threshold \(\bar{u}_{T_1+t}\) for the first time at time \(T_1+t\), subsequent to time \(T_1\), regardless of whether \(m_n(t)\) had previously surpassed \(u_t\) before time \(T_1\). Given that the sequences $\{u_t(T_1)\}_{t=1}^{T_1}$ and $\{u_t(T_1)\}_{t=T_1+1}^{T}$ decrease separately, the events $E_t$ for $t=1,\ldots,T_1$, $\bar{E}_{T_1}$, $E_{T_1+\tau}$ for $\tau=1,\ldots,T-T_1$, and $\bar{E}_{T}$ together form the complete event space for the total process. Then social welfare (\ref{socialobj}) can be written as\begin{align*}
    \sum^T_{t=1}\sum^N_{n=1}\mathbb{E}[r_{a_n(t)}]=&\sum^N_{n=1}\sum^{T_1}_{t=1}\mathbb{E}[r_{a_n(t)}]+\sum^N_{n=1}\sum^{T}_{t=T_1+1}\mathbb{E}[r_{a_n(t)}]\\
    =&\sum^N_{n=1}\sum^{T_1}_{\tau=1}\mathbb{P}(E_{\tau})\sum^{T_1}_{t=1}\mathbb{E}[r_{a_n(t)}|E_{\tau}]+\sum^N_{n=1}\mathbb{P}(\bar{E}_{T_1})\sum^{T_1}_{t=1}\mathbb{E}[r_{a_n(t)}|\bar{E}_{T_1}]+\sum^N_{n=1}\sum^{T-T_1}_{\tau=1}\mathbb{P}(E_{T_1+\tau})\sum^{T}_{t=T_1+1}\mathbb{E}[r_{a_n(t)}|E_{T_1+\tau}]\\
    &+\sum^N_{n=1}\mathbb{P}(\bar{E}_{T})\sum^{T}_{t=T_1+1}\mathbb{E}[r_{a_n(t)}|\bar{E}_{T}].\end{align*}
Conditioned on the event $E_{\tau}$, each agent will explore for $\tau$ time slots and then exploit for the remaining $(T_1-\tau+1)$ time slots. Similarly, conditioned on the event $E_{T_1+\tau}$, each agent will explore for $(\tau-1)$ time slots and exploit for the subsequent $(T-T_1-\tau+1)$ time slots. Since the process is the same for all agents, we omit the superscript \(n\) in \(R^n_i\) in the following. Then 
\begin{align*}
    =&N\sum^{T_1}_{\tau=1}\mathbb{P}(E_{\tau})\left(\mu \tau+\mathbb{E}[\max_{i \in\{0,\ldots,i-1\}} R_{i}|E_{\tau}](T_1-\tau+1)\right)+N\mathbb{P}(\bar{E}_{T_1})\mu(T_1+1)+N\mathbb{P}(E_{T_1+1})\mathbb{E}[\max_{n\in\mathcal{N}}\{m_n(T_1)\}|E_{T_1+1}](T-T_1)\\
    &+N\sum^{T-T_1-1}_{\tau=1}\mathbb{P}(E_{T_1+\tau+1})\left(\mu \tau+\mathbb{E}[\max_{n\in\mathcal{N},\tau\in\{1,\ldots,\tau\}}\{m_n(T_1),R_{T_1+\tau}\}|E_{T_1+\tau+1}](T-T_1-\tau)\right)+N\mathbb{P}(\bar{E}_{T})\mu(T-T_1)\\
    =&N\sum^{T_1}_{\tau=1}\left(\left(F(u_{\tau-1}(T_1))^{\tau-1}-F(u_\tau(T_1))^{\tau}\right)\mu \tau+\left(\int^{u_{\tau-1}(T_1)}_{u_\tau(T_1)}rdF(r)^{\tau}+F(u_{\tau-1}(T_1))^{\tau-1}\int^1_{u_{\tau-1}(T_1)}rdF(r)\right)(T_1-\tau+1)\right)\\
    &+NF(u_{T_1}(T_1))^{T_1}\mu(T_1+1)+N\int^1_{\bar{u}_{T_1+1}}rdP(r)^{N}(T-T_1)+N\sum^{T-T_1-1}_{\tau=1}\bigg(\left(P(\bar{u}_{T_1+\tau})^{N}F(\bar{u}_{T_1+\tau})^{\tau-1}-P(\bar{u}_{T_1+\tau+1})^{N}F(\bar{u}_{T_1+\tau+1})^{\tau}\right)\mu \tau\\
    &\left.+\left(\int^{\bar{u}_{T_1+\tau}}_{\bar{u}_{T_1+\tau+1}}rdP(r)^NF(r)^{\tau}+P(\bar{u}_{T_1+\tau})^NF(\bar{u}_{T_1+\tau})^{\tau-1}\int^1_{\bar{u}_{T_1+\tau}}rdF(r)\right)(T-T_1-\tau)\right)+NP(\bar{u}_T)^{N}F(r)^{T-T_1-1}\mu(T-T_1),\\
\end{align*}
where we set $u_0(T_1) = 1$ for simplification, which does not affect the agents' actions at time $0$ since agents can only choose to explore. Then
\begin{align*}
    =&N\sum^{T_1}_{\tau=0}F(u_{\tau}(T_1))^{\tau}\mu+N\sum^{T_1}_{\tau=1}\left(\int^{u_{\tau-1}(T_1)}_{u_{\tau}(T_1)}rdF(r)^{\tau}+F(u_{\tau-1}(T_1))^{\tau-1}\int^1_{u_{\tau-1}(T_1)}rdF(r)\right)(T_1-\tau+1)\\
    &+N\int^1_{\bar{u}_{T_1+1}}rdP(r)^{N}(T-T_1)+N\sum^{T-T_1-1}_{\tau=1}\bigg(\left(P(\bar{u}_{T_1+\tau})^{N}F(\bar{u}_{T_1+\tau})^{\tau-1}-P(\bar{u}_{T_1+\tau+1})^{N}F(\bar{u}_{T_1+\tau+1})^{\tau}\right)\mu \tau\\
    &+\bigg(P(\bar{u}_{T_1+\tau})^NF(\bar{u}_{T_1+\tau})^{\tau-1}\left(1-\int^1_{\bar{u}_{T_1+\tau}}F(r)dr\right)-\bar{u}_{T_1+\tau+1}P(\bar{u}_{T_1+\tau+1})^NF(\bar{u}_{T_1+\tau+1})^{\tau}\\
    &-\int^{\bar{u}_{T_1+\tau}}_{\bar{u}_{T_1+\tau+1}}P(r)^NF(r)^tdr\bigg)(T-T_1-\tau)\bigg)+NP(\bar{u}_T)^{N}F(\bar{u}_T)^{T-T_1-1}\mu(T-T_1).\\
    \end{align*}
    Replace $\left(1-\int^1_{\bar{u}_{T_1+\tau}}F(r)dr\right)(T-T_1-\tau)$ in the last sentence by the part in equation (\ref{singlethreshold}) as
    \begin{align*}
        \left(1-\int^1_{\bar{u}_{T_1+\tau}}F(r)dr\right)(T-T_1-\tau)=\bar{u}_{T_1+\tau}(T-T_1-\tau+1)-\mu,
    \end{align*}
then we have
    \begin{align*}
    =&N\sum^{T_1}_{\tau=0}F(u_{\tau}(T_1))^{\tau}\mu+N\sum^{T_1}_{\tau=1}\left(\int^{u_{\tau-1}(T_1)}_{u_{\tau}(T_1)}rdF(r)^{\tau}+F(u_{\tau-1}(T_1))^{\tau-1}\int^1_{u_{\tau-1}(T_1)}rdF(r)\right)(T_1-\tau+1)\\
    &+N\int^1_{\bar{u}_{T_1+1}}rdP(r)^{N}(T-T_1)+N\sum^{T-T_1-1}_{\tau=1}\bigg(\left(P(\bar{u}_{T_1+\tau})^{N}F(\bar{u}_{T_1+\tau})^{\tau-1}-P(\bar{u}_{T_1+\tau+1})^{N}F(\bar{u}_{T_1+\tau+1})^{\tau}\right)\mu \tau-\bigg(\bar{u}_{T_1+\tau+1}P(\bar{u}_{T_1+\tau+1})^NF(\bar{u}_{T_1+\tau+1})^\tau\\
    &+\int^{\bar{u}_{T_1+\tau}}_{\bar{u}_{T_1+\tau+1}}P(r)^NF(r)^\tau dr\bigg)(T-T_1-\tau)+P(\bar{u}_{T_1+\tau})^NF(\bar{u}_{T_1+\tau})^{\tau-1}\left(\bar{u}_{T_1+\tau}(T-T_1-\tau+1)-\mu\right)\bigg)+NP(\bar{u}_T)^{N}F(\bar{u}_T)^{T-T_1-1}\mu(T-T_1)\\
        \stackrel{\text{(a)}}{=}&N\sum^{T_1}_{\tau=0}F(u_{\tau}(T_1))^{\tau}\mu+N\sum^{T_1}_{\tau=1}\left(\int^{u_{\tau-1}(T_1)}_{u_{\tau}(T_1)}rdF(r)^{\tau}+F(u_{\tau-1}(T_1))^{\tau-1}\int^1_{u_{\tau-1}(T_1)}rdF(r)\right)(T_1-\tau+1)\\
    &+N\int^1_{u_{T_1+1}}rdP(r)^{N}(T-T_1)+NP(\bar{u}_{T_1+1})^Nu_{T_1+1}(T-T_1)-N\sum^{T-T_1-1}_{\tau=1}\int^{\bar{u}_{T_1+\tau}}_{\bar{u}_{T_1+\tau+1}}P(r)^NF(r)^tdr(T-T_1-\tau)\\
    =&N\sum^{T_1}_{\tau=0}F(u_{\tau}(T_1))^{\tau}\mu+N\sum^{T_1-1}_{\tau=0}\left(\int^{u_{\tau}(T_1)}_{u_{\tau+1}(T_1)}rdF(r)^{\tau+1}+F(u_{\tau}(T_1))^{\tau}\int^1_{u_{\tau}(T_1)}rdF(r)\right)(T_1-\tau)\\
    &+N\left(1-\int^1_{\bar{u}_{T_1+1}}P(r)^{N}dr\right)(T-T_1)-N\sum^{T-T_1-1}_{\tau=1}\int^{\bar{u}_{T_1+\tau}}_{\bar{u}_{T_1+\tau+1}}P(r)^NF(r)^{\tau}dr(T-T_1-\tau),
    \end{align*}
where equation (a) utilizes the fact that \(\bar{u}_T = \mu\). Then we complete the proof of (\ref{eqVN}).

Let $\tau$, in $\{1,\ldots,T+1\}$, denote the random variable representing exploration time by non-myopic agents under communication restriction mechanism $\mathcal{M}_2=[T]\setminus\{T_1\}$ in the total process. Then expectation of $\tau$ is
\begin{align*}
\mathbb{E}[\tau]=&\sum^{T+1}_{t=1}t\mathbb{P}(\tau=t)=\sum^{T+1}_{t=1}\mathbb{P}(\tau\geq t)=1+\sum^{T_1+1}_{t=2}F(u_{t-1})^{t-1}+\sum^{T+1}_{t=T_1+2}P(u_{t-1})^NF(u_{t-1})^{t-T_1-1}\\
=&1+\sum^{T_1}_{t=1}F(u_{t})^{t}+\sum^{T}_{t=T_1+1}P(u_{t})^NF(u_{t})^{t-T_1}
\end{align*}
which completes the proof of (\ref{T1time}).

\section{Proof of Theorem \ref{T2}}\label{Appendix10}
The sequence \(\{u_t(T_1)\}_{t=1}^T\) cannot be expressed in a closed-form solution and must be determined using numerical methods based on equations (\ref{Mt<T_1+1}) and (\ref{singlethreshold}). We rewrite equations (\ref{Mt<T_1+1}) and (\ref{singlethreshold}) as follows, and the solution \(\{u_t(T_1)\}_{t=1}^T\) corresponds to the roots of these functions:
\begin{align*}
&g_t(u_1(T_1),\ldots,u_{T_1}(T_1))=\\
&\left\{\begin{aligned}
    &u_t(T_1)-\mu-(T_1-t)\left(\int^1_{u_t(T_1)} 1-F(r)dr\right)-(T-T_1)\int^1_{u_t(T_1)}P(r)^{N-1}(1-F(r))dr,\text{if }u_t(T_1)>u_{T_1+1}(T_1),\\
&u_t(T_1)-\mu-(T_1-t)\left(\int^1_{u_t(T_1)} 1-F(r)dr\right)-(T-T_1)\int^{1}_{\bar{u}_{T_1+1}}P(r)^{N-1}(1-F(r))dr+\sum^{k-1}_{j=1}(T-T_1-j)\int^{\bar{u}_{T_1+j}}_{\bar{u}_{T_1+j+1}}P(r)^{N-1}F(r)^{j}(1-F(r))dr\nonumber\\
    &-(T-T_1-k)\int^{\bar{u}_{T_1+k}}_{u_t(T_1)}P(r)^{N-1}F(r)^{k}(1-F(r))dr,\;\;\qquad\qquad\qquad\qquad\qquad\qquad\text{if } \bar{u}_{T_1+k}>u_t(T_1)>\bar{u}_{T_1+k+1},t=1,\ldots,T_1,
\end{aligned}
\right.\\
&\bar{g}_t(\bar{u}_t)=\bar{u}_t-\mu-(T-t)\left(\int^1_{\bar{u}_t} 1-F(r)dr\right), t=1,\ldots,T.
\end{align*}

Define $\mathbf{u}(T_1) = (u_1(T_1), \ldots, u_{T_1}(T_1))$ and $\mathbf{g}(\cdot) = (g_1(\cdot), \ldots, g_{T_1}(\cdot))$. Our algorithm first obtains \(\{\bar{u}_t\}_{t=1}^T\) by finding the roots of \(T\) equations \(\bar{g}_t(\cdot)\) using Newton's method, then takes it as the initial solution for the Newton's method for $\{u_t(T_1)\}_{t=1}^{T_1}$. Solving $\{u_t(T_1)\}_{t=1}^{T_1}$ requires the Jacobian matrix $\mathbf{J}(\cdot)=(J_{ij}(\cdot)),i,j\in\{1,\ldots,T_1\}$ of $\mathbf{g}(\cdot)$:
\begin{align*}
    J_{tt}(\mathbf{u}(T_1)) =& \begin{cases}
        1 + (1 - F(u_t(T_1)))\left((T_1 - t) + (T - T_1)P(u_t(T_1))^{N-1}\right),&\text{if } u_t(T_1) > u_{T_1+1}(T_1)\\
        1 + (1 - F(u_t(T_1)))\left((T_1 - t) + (T - T_1 - i)P(u_t(T_1))^{N-1}F(u_t(T_1))^i\right),& \text{if } u_{T_1+i}(T_1) > u_t(T_1) > u_{T_1+i+1}(T_1) \\
   \end{cases}\\
   J_{ti} =&0, t \neq i\label{J},
\end{align*}
which only has non-zero entries on the diagonal. We provide the Newton's method for $\{u_t(T_1)\}_{t=1}^{T_1}$ in steps 4-5 in Alg. \ref{A2} as follows:
\begin{itemize}
    \item[a] initialize solution $\mathbf{u}^{(0)}(T_1)=\{\bar{u}_t\}_{t=1}^{T_1}$, $k$=0;
    \item[b] Solve $\Delta\mathbf{u}^{(k)}(T_1)$ by equation $\mathbf{J}(\mathbf{u}^{(k)}(T_1))\Delta\mathbf{u}^{(k)}(T_1)=-\mathbf{g}(\mathbf{u}^{(k)}(T_1))$;
    \item[c] update the solution as $\mathbf{u}^{(k+1)}(T_1)=\mathbf{u}^{(k)}(T_1)+\Delta\mathbf{u}^{(k)}(T_1)$;
    \item[d] Repeat the iteration, and end the process when $|\mathbf{g}(\mathbf{u}^{(k)}(T_1))|<\epsilon$.
\end{itemize}

The primary computational complexity of Newton's method stems from solving the linear equation for \(\Delta\mathbf{u}^{(k)}(T_1)\) in step b, resulting in a complexity of \(O(T)\) for our problem. Given a specific \(\epsilon\), the complexity of the iteration step remains constant, independent of \(T_1\). Therefore, when scanning through all possible \(T_1\) values to identify the optimal one-time communication mechanism \(\mathcal{M}_2^* = [T] \setminus \{T_1^*\}\), the overall complexity escalates to \(O(T^2)\).

\end{document}